\newif\ifllncs
\llncsfalse

\documentclass{article}
\usepackage{mft}

\usepackage{enumitem} 

\newcommand{\changed}[1]{{\textcolor{blue}{#1}}}

\newcommand{\chlspace}{\mathscr{C}}
\newcommand{\keyspace}{\mathscr{K}}

\newcommand{\Reveal}{\mathsf{Reveal}}
\newcommand{\Exp}{\mathsf{Exp}}

\newcommand{\prover}{\ensuremath{\mathcal{P}}}
\newcommand{\verify}{\ensuremath{\mathcal{V}}}

\newcommand{\starxi}{\ensuremath{\xi}}
\newcommand{\starprover}{\ensuremath{\prover}}
\newcommand{\starverify}{\ensuremath{\verify}}

\newcommand{\tvproj}{\ensuremath{\Pi^{T_V}}}
\newcommand{\accpovm}{\ensuremath{M^\mathrm{acc}}}
\newcommand{\rejpovm}{\ensuremath{M^\mathrm{rej}}}

\title{On Removing Interaction from Quantum Proofs}
\author{Nicholas Spooner \\ Cornell University \and Max Tromanhauser \\ Cornell University}
\date{}

\begin{document}

\maketitle
\thispagestyle{plain}

\begin{abstract}
	An important open question in quantum cryptography is the construction of publicly-verifiable NIZKs for QMA. Classically, one can construct NIZKs for NP in the random oracle model (and sometimes in the standard model) by compiling an honest-verifier ZK (HVZK) $\Sigma$-protocol for NP using the Fiat--Shamir transformation. Broadbent and Grilo introduced a quantum analog of a $\Sigma$-protocol (which they call a $\Xi$-protocol) in which the prover's first message is quantum, and show that HVZK $\Xi$-protocols exist for QMA. However, it is not clear how to compile such protocols into NIZKs in the (Q)ROM, because the Fiat--Shamir transformation seems to be incompatible with quantum messages. In this work we give formal evidence that this is indeed the case: we show that if generic ``Fiat--Shamir-like'' compilers for quantum protocols exist in the QROM (with small completeness and soundness error) then $\QMA = \BQP$.
\end{abstract}

\section{Introduction}
\label{sec:intro}

This work is motivated by the problem of constructing \emph{nontrivial publicly-verifiable non-interactive (quantum) proof systems} for $\QMA$.
A proof system for $\QMA$ is a protocol in which a prover aims to convince a verifier that a given input $x$ belongs to a language $\lang \in \QMA$.
By {\em nontrivial} we mean that the system satisfies some useful property, like zero knowledge or succinctness, which is not achieved by simply sending the $\QMA$ witness.
A proof system is \emph{non-interactive} if the communication consists of a single message from the prover.

One common approach to designing non-interactive proof systems is to start from an \emph{interactive} proof system and then remove the interaction using a generic compilation.
This is most commonly achieved using the \emph{Fiat--Shamir (FS) transformation} \cite{FiatS86}, which applies to classical public-coin interactive proofs (i.e., where the verifier only sends uniformly random strings).
In a nutshell, the FS transformation has the prover simulate each verifier message by applying a cryptographic hash function $h$ to the protocol transcript so far.
Provided the hash function is sufficiently ``unpredictable,'' this forces the prover to commit to all of its prior messages before learning the verifier's next message.

The FS transformation can be proven secure if $h$ is modelled as a (quantum-accessible) random oracle.
By applying the transformation to public-coin zero knowledge proofs and succinct arguments, one can obtain non-interactive zero knowledge (NIZK) arguments and succinct non-interactive arguments (SNARGs) in the random oracle model (ROM).
While it is known that the Fiat--Shamir transformation cannot be proven secure in the standard model in general \cite{GoldwasserK03}, for certain protocols one can instantiate $h$ with a well-chosen cryptographic hash function to obtain a secure scheme in the standard model \cite{CanettiCHLRR18}.\footnote{This is a very substantial line of work and so we do not attempt to provide a full list of citations.}

Unfortunately we do not know how to leverage the FS transformation to obtain publicly-verifiable NIZKs or SNARGs for $\QMA$ even in the ROM.
The reason is that known interactive proof systems for $\QMA$ are either based on protocols for classical verification of quantum computations, which have classical communication but require the verifier to hold private state, or are public-coin but with quantum prover messages \cite{BroadbentG22,GunnJMZ22}.
In the former case, one can obtain \emph{designated-verifier} NIZKs/SNARGs for $\QMA$ in the ROM by leveraging fully-homomorphic encryption \cite{AlagicCGH20,BartusekKLMMVVY22}; it is unclear whether this approach can lead to publicly verifiable NIZKs/SNARGs.

For protocols with quantum communication, one possible adaptation of classical Fiat--Shamir might be to coherently compute a classical hash of the first message and measure the output. This is problematic for two reasons.
First, the resulting protocol may not be \emph{complete}, because the measurement may substantially disturb the quantum message and cause the verifier to reject.
Second, the resulting protocol may not be \emph{sound}, because (e.g.) if the verifier checks the hash in the $Z$ basis, the adversary can apply arbitrary $Z$ operations without affecting that check.
(These issues are related to the impossibility of signing quantum states \cite{AlagicGM21}.)

In this work we seek to determine whether this difficulty represents a barrier to \emph{any} generic technique for removing interaction from quantum protocols.
More precisely, we ask the question:
\begin{center}
  {\em Is there a generic method for removing interaction from protocols with quantum communication?}
\end{center}

\subsection{Our Results}
\label{sec:results}

This work initiates the study of that question by providing a partial negative answer.

First, we formalize what we mean by a ``generic method for removing interaction'' by defining an object we call a (Quantum) Non-Interactivity Compiler ((Q)NIC).
Informally, a QNIC is a black-box compilation that converts an interactive $\Sigma$-protocol with quantum communication (also known as a $\Xi$-protocol \cite{BroadbentG22}) into a publicly-verifiable non-interactive quantum argument.
We additionally define a subclass of {\em straightline} QNICs which capture compilation strategies that require only a single black-box interaction with the interactive proof.
The FS transformation is an example of a straightline classical NIC.

Our main contribution is to rule out the existence of straightline QNICs with small completeness and soundness error assuming the existence of quantum commitments\footnote{
  Note that quantum zero knowledge for any hard-on-average language implies quantum commitments \cite{BrakerskiCQ23}.
}.

\begin{theorem}[Informal]
  Let $\oracledist$ be an efficiently simulatable distribution over unitary oracles. Assume that
  \begin{itemize}
    \item there exist quantum commitments with respect to $\oracledist$, and
    \item there exist $(1 - \varepsilon)$-complete, $\delta$-sound, straightline QNICs relative to $\oracledist$ where
      $\varepsilon(n) + 2\delta(n) \leq \frac{3}{5}$.
  \end{itemize}
  Then $\QMA = \BQP$.
  \label{thm:main-result}
\end{theorem}

Our full result (stated in \cref{sec:impossibility}) covers distributions over arbitrary unitary oracles, although the class collapsing implication is less straightforward in the general case.
It is important that the result relativizes because Fiat--Shamir is not secure in the plain model in general, so any meaningful impossibility should hold in, for example, the (quantum) random oracle model.
In particular, \cref{thm:main-result} implies corollaries in the following well-studied settings.
\begin{itemize}
  \item In the plain model (i.e., $\oracledist$ is an empty distribution) or the common reference string (CRS) model (i.e., where $\oracle_\mathsf{crs}\ket{x} = \ket{x \oplus \mathsf{crs}}$ for $\oracle_\mathsf{crs} \gets \oracledist$), assuming the existence of both quantum commitments and small-error straightline QNICs implies $\QMA = \BQP$.

  \item In the QROM or the quantum Haar random oracle model (QHROM), quantum commitments exist unconditionally. In these cases the implication $\QMA = \BQP$ follows from just the existence of small-error straightline QNICs relative to those models.
\end{itemize}

We require that $\oracledist$ be a distribution over unitary oracles because our decision algorithm requires access to the oracle's inverse.
Thankfully this is not an onerous requirement and it is satisfied by most relevant oracle models
(in particular, any quantum-accessible classical oracle).

\paragraph{Open Questions.} Our result represents a first step towards understanding whether interaction can be generically removed from quantum protocols. There remains plenty of scope for future work:
\begin{itemize}
    \item It remains open whether our impossibility result can be extended to a broader class of compilers which either are not restricted to straightline behavior or have larger completeness and soundness error. Note that while the completeness and soundness error of a QNIC can be generically reduced via parallel repetition, the resulting protocol is no longer straightline and so our result does not apply. Note also that there do exist classical NICs which are not straightline, and so could potentially be used in the quantum setting (e.g., the Fischlin transform~\cite{Fischlin05}, which requires multiple interactions with the prover to generate a proof).
    \item Our result applies only to QNICs which yield publicly-verifiable protocols; it remains an intriguing open question to design QNICs that yield designated-verifier protocols.
    \item Our result rules out QNICs that work for any $\Xi$-protocol; potentially there exist QNICs that work for a subclass of protocols satisfying some constraint.
\end{itemize}

There may also be other ways to obtain non-interactive proof systems for $\QMA$ that do not look like removing interaction (analogously to e.g. the hidden bits model in the classical setting).

\subsection{Related Work}

Classically, the FS transformation has been proven secure in the ROM \cite{BellareR93,PointchevalS96}, and a long line of work was required to upgrade this to the QROM.
Several works \cite{DagdelenFG13,AmbainisRU14} demonstrated barriers for standard proof techniques, while others demonstrated that particular modifications of the transformation were secure \cite{Unruh15,Unruh17,KiltzLS18}.
The standard Fiat--Shamir was eventually proven secure in the QROM in \cite{DonFMS19,LiuZ19}.
These results all concern the transformation applied to classical $\Sigma$-protocols which do not support quantum communication.
Another line of work has proven that the classical FS transformation is not secure in the plain model in general \cite{GoldwasserK03}.
In contrast, our result is agnostic to the oracle model and so covers (e.g.) the QROM.

Several strategies for developing nontrivial proof systems for $\QMA$ have appeared in the literature.
It has been shown that there exist public-coin, three-message succinct arguments \cite{GunnJMZ22} (assuming the quantum PCP conjecture) and zero knowledge proofs \cite{BroadbentG22} for $\QMA$ if we allow the prover to send quantum messages.
If we wish to avoid interaction, the main approaches have been to use a designated-verifier \cite{AlagicCGH20,BartusekKLMMVVY22,ColadangeloVZ20,Shmueli21,MorimaeY22} --- i.e., a verifier that receives a secret setup with which it can confirm the validity of the proof --- or to rely on heuristic security of post-quantum obfuscation \cite{BartusekM22,BartusekJRT25}.

\subsection{Technical Overview}
\label{sec:tech-overview}

Before we attempt to prove our main result, we must formalize the class of transformations which we are considering.
At a high level, we'd like to capture strategies which resemble the structure of \cref{fig:intro-compiler} --- i.e., those which, through a straightline interaction with the prover, create a non-interactive quantum proof that can be unfolded into a valid transcript.
This model consists of two algorithms: $\mathcal{T}_P$ which, through a single interaction with the prover, generates a quantum proof; and $\mathcal{T}_V$ which confirms the validity of the proof and outputs a transcript consistent with the compiled protocol.
The proof is accepted if $\mathcal{T}_V$'s output is accepted by the verifier\footnote{
  Our model actually captures strategies which are slightly more general.
  Rather than having the non-interactive verifier forward $\verify$'s output, we allow $\mathcal{T}_V$ to make a final decision after receiving $\verify$'s verdict.
  See \cref{sec:model} for details.
}.
Together $\mathcal{T}_P$ and $\mathcal{T}_V$ define a straightline QNIC, and we require them to be ``universal'' in the sense that, for any secure $\Xi$-protocol, the resulting non-interactive argument is also secure.

\begin{figure}
  \centering
  \begin{tikzpicture}[
    box/.style={draw, thick, minimum width=3em, minimum height=8em, align=center},
    msg/.style={->, thick, shorten <=0.5em, shorten >=0.5em},
    label/.style={midway, fill=white, inner sep=3pt}
  ]
    \def\padding{0.5em}

    \node[box] (P) {$\prover$};
    \node[box, right=5em of P] (TP) {$\mathcal{T}_P$};
    \node[box, right=5em of TP] (TV) {$\mathcal{T}_V$};
    \node[box, right=5em of TV] (V) {$\verify$};

    \draw[msg] ($(P.east)+(0, 3em)$) -- ($(TP.west)+(0, 3em)$) node[label, above] {$\reg{A}$};
    \draw[msg] ($(TP.west)+(0, 0mm)$) -- ($(P.east)+(0, 0mm)$) node[label, above] {$c$};
    \draw[msg] ($(P.east)+(0,-3em)$) -- ($(TP.west)+(0,-3em)$) node[label, above] {$\reg{Z}$};

    \draw[msg] (TP.east) -- (TV.west) node[label, above] {$\pi$};

    \draw[msg] ($(TV.east)+(0, 3em)$) -- ($(V.west)+(0, 3em)$) node[label, above] {$\reg{A'}$};
    \draw[msg] ($(TV.east)+(0, 0mm)$) -- ($(V.west)+(0, 0mm)$) node[label, above] {$c'$};
    \draw[msg] ($(TV.east)+(0,-3em)$) -- ($(V.west)+(0,-3em)$) node[label, above] {$\reg{Z'}$};

    \draw[thick] (V.south)  -- ++(0,-1.2em) node[below] {$0/1$};

    \draw[dashed,semithick] ($(P.west)+(-1em,5em)$)
      -- ($(TP.east)+(1em,5em)$)
      -- ($(TP.east)+(1em,-5em)$)
      -- ($(P.west)+(-1em,-5em)$)
      -- cycle;
    \node[anchor=west] at ($(P.west)+(-1em,6em)$) {Non-Interactive Prover};

    \draw[dashed,semithick] ($(TV.west)+(-1em,5em)$)
      -- ($(V.east)+(1em,5em)$)
      -- ($(V.east)+(1em,-5em)$)
      -- ($(TV.west)+(-1em,-5em)$)
      -- cycle;
    \node[anchor=west] at ($(TV.west)+(-1em,6em)$) {Non-Interactive Verifier};
  \end{tikzpicture}
  \caption{Straightline ``Fiat--Shamir-like'' compiler (where $\reg{A}$, $\reg{Z}$, $\reg{A'}$, and $\reg{Z'}$ are quantum registers).}
  \label{fig:intro-compiler}
\end{figure}
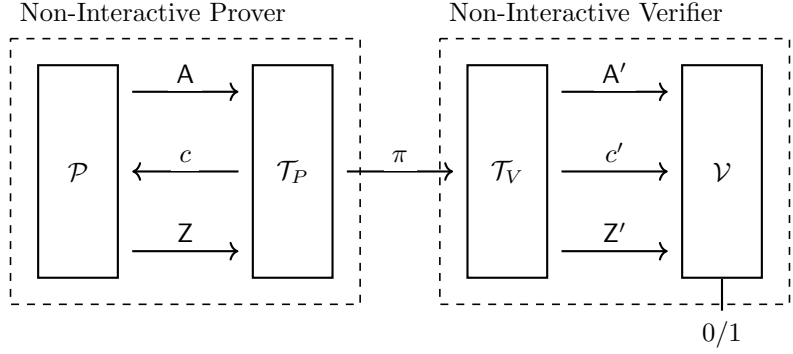

Our strategy towards proving \cref{thm:main-result} is to assume that such a $(\mathcal{T}_P, \mathcal{T}_V)$ exist and leverage them to construct a decision algorithm for a $\QMA$-complete language $\lang$.
As an intermediate objective, we construct a family of $\Xi$-protocols, each of which securely proves membership in $\lang$, but if compiled with $(\mathcal{T}_P, \mathcal{T}_V)$, the resulting NIZK will be vulnerable to a forgery attack.
That attack is then used directly to construct the decision algorithm for $\lang$.

The key insight for this counterexample family comes from understanding the security of the classical FS transformation.
At its heart, the transform enforces a dependence between the initial message $a$ and the challenge $c$.
That dependence prevents a cheating prover from correlating their message with a favorable challenge.
Any prover that succeeds in the non-interactive argument must behave as though it had committed to $a$ before learning $c$, thereby inheriting the soundness of the interactive protocol.

As we expand our ambitions to include quantum communication, introducing a correlation between a quantum message and a classical challenge introduces a complication: it requires a measurement be performed on the message state.
Measurements are necessarily destructive and may irreversibly disturb the message in a way which compromises its completeness.
This is the tension which we exploit to build our counterexample.

Our construction starts with an existing $\Xi$-protocol $(\prover, \verify)$ (e.g., from \cite{BroadbentG22}) which we will modify. 
The QNIC algorithms $(\mathcal{T}_P, \mathcal{T}_V)$ together can be understood as a man-in-the-middle attempting to make some undetectable measurement on the first message.
Given this framing, our goal is to modify $\prover$ and $\verify$ such that any man-in-the-middle adversary must either pass the first message through unaltered or else $\verify$ will reject.

\paragraph{Quantum Ciphertext Authentication.}
Introduced in \cite{AlagicGM18}, quantum ciphertext authentication (QCA) provides a natural tool for enforcing this behavior.
Any eavesdropper listening in on a QCA-secure quantum encryption scheme will be unable to alter or forge a ciphertext without being caught.
The security property is defined through a simulation-based definition where an unbounded adversary is given a single copy of either a `real' or `simulated' ciphertext and it attempts to distinguish between the cases.

Formally, the `real world' is modeled as a quantum channel $\Gamma = \E_{k}[\Dec_k \circ \mathcal{A} \circ \Enc_k]$ (sometimes referred to as the attacker's `effective attack map') which composes encryption, the fixed man-in-the-middle attack $\mathcal{A}$, and then decryption.
The channel's input message can be adversarially chosen by $\mathcal{A}$ and entangled with an outside register $\mathcal{A}$ controls.
In the simulated world, we have a corresponding channel $\Lambda$ where $\Enc$ and $\Dec$ are replaced with `ideal' versions:
\begin{itemize}
\item The ciphertext provided by the simulated $\Enc$ does not contain the adversary's message but instead encrypts a state which is maximally entangled with a register private to $\Enc$ and $\Dec$.
\item The simulated $\Dec$ then validates the ciphertext and confirms that the entanglement between the message register and the private register has been preserved.
\end{itemize}
If both conditions are met, $\Lambda$ outputs the original message.
An encryption scheme is QCA-secure if $\Gamma$ and $\Lambda$ are negligibly close for all $\mathcal{A}$.

Using a QCA scheme as a building block, our counterexample family is constructed by modifying the $\Xi$-protocol as follows.
The prover's initial message is encrypted under a key $k$ which is part of the protocol description.
The verifier sends a uniform challenge as before, and the prover's final message is similarly unchanged.
The verifier attempts to decrypt the first message register $\reg{A}$ using $k$.
If this fails, then the verifier rejects; otherwise, it verifies the transcript incorporating the decrypted message.

Note that this construction yields a \emph{family} of protocols: $\prover_k$ and $\verify_k$ for each key $k$. Our eventual decision algorithm will compile a randomly-chosen protocol from this family. Importantly, we are not modifying the model to have the prover and verifier share a key: $k$ is part of the \emph{description} of the protocol. The role of $k$ is to exploit the black-box nature of $(\mathcal{T}_P,\mathcal{T}_V)$: the behavior of these algorithms must be independent of $k$, and so we will be able to obtain security guarantees against them.

\paragraph{Attacking the NIZK.}
Notice that in the compilation of these protocols the QNIC algorithms $(\mathcal{T}_P, \mathcal{T}_V)$ never see the prover's first message in the clear.
Instead, $\mathcal{T}_P$ queries $\prover_k$ which provides a message which has already been encrypted.
If we switch to a simulated $\Enc$ (as in the QCA definition), the message is placed aside and remains untouched until it is decrypted.
In this way, we should be able to use the simulated $\Enc$ and $\Dec$ to `defer' the choice of message until the moment of decryption, which in our setting is when the unrolled transcript is given to $\verify_k$.
At that point, we would like to swap the simulated ciphertext for a real ciphertext with some freshly chosen message which depends on the challenge $c$.

\ifllncs
\begin{figure}
  \centering
  \fbox{%
    \begin{minipage}[t]{\dimexpr\textwidth-2\fboxsep-2\fboxrule\relax}
      $\mathsf{Real}(\rho_\reg{A} \otimes \rho_\reg{B})$
      \begin{enumerate}
        \item Sample key $k \gets \Gen(1^n)$ and encryption randomness $r$.
        \item Prepare $\sigma =  \mathcal{A} \circ (\Enc_{k;r} \otimes \Id_\reg{B})(\rho_\reg{A} \otimes \rho_\reg{B})$.
        \item Measure if $\reg{A}$ is a valid ciphertext.
        \item If yes, output the residual state.
        \item Otherwise, output $\bot \otimes \sigma_\reg{B}$.
      \end{enumerate}
    \end{minipage}
  }

  \vspace{2mm}

  \fbox{%
    \begin{minipage}[t]{\dimexpr\textwidth-2\fboxsep-2\fboxrule\relax}
      $\mathsf{Defer}(\rho_\reg{A} \otimes \rho_\reg{B})$
      \begin{enumerate}
        \item Sample key $k \gets \Gen(1^n)$ and encryption randomness $r$.
        \item Prepare $\sigma = \mathcal{A} \circ (\Enc_{k;r} \otimes \Id_\reg{B})(\Phi^+_\reg{ZA} \otimes \rho_\reg{B})$.
        \item Measure if $\reg{A}$ is a valid ciphertext encrypting half of $\ket{\Phi^+}$.
        \item If yes, output $\Enc_{k;r}(\rho_\reg{A}) \otimes \sigma_\reg{B}$. \label{item:intro-exp-defer-yes}
        \item Otherwise, output $\bot \otimes \sigma_\reg{B}$.
      \end{enumerate}
    \end{minipage}
  }
  \caption{`Real' and `deferring' channels (where $\bot$ is some failure state).}
  \label{fig:intro-experiments}
\end{figure}
\else
\begin{figure}
  \centering
  \fbox{%
    \begin{minipage}[t][11em]{0.4\linewidth}
      $\mathsf{Real}(\rho_\reg{A} \otimes \rho_\reg{B})$
      \begin{enumerate}
        \item Sample $k \gets \Gen(1^n)$.
        \item Sample encryption randomness $r$.
        \item Prepare
          {
            \setlength{\abovedisplayskip}{2pt}
            \setlength{\belowdisplayskip}{2pt}
            \[
              \sigma =  \mathcal{A} \circ (\Enc_{k;r} \otimes \Id_\reg{B})(\rho_\reg{A} \otimes \rho_\reg{B}).
            \]
          }
        \item Measure if $\reg{A}$ is a valid ciphertext.
        \item If yes, output the residual state.
        \item Otherwise, output $\bot \otimes \sigma_\reg{B}$.
      \end{enumerate}
    \end{minipage}
  }
  \fbox{%
    \begin{minipage}[t][11em]{0.4\linewidth}
      $\mathsf{Defer}(\rho_\reg{A} \otimes \rho_\reg{B})$
      \begin{enumerate}
        \item Sample $k \gets \Gen(1^n)$.
        \item Sample encryption randomness $r$.
        \item Prepare
          {
            \setlength{\abovedisplayskip}{2pt}
            \setlength{\belowdisplayskip}{2pt}
            \[
              \sigma = \mathcal{A} \circ (\Enc_{k;r} \otimes \Id_\reg{B})(\Phi^+_\reg{ZA} \otimes \rho_\reg{B}).
            \]
          }
        \item Measure if $\reg{A}$ is a valid ciphertext encrypting half of $\ket{\Phi^+}$.
        \item If yes, output $\Enc_{k;r}(\rho_\reg{A}) \otimes \sigma_\reg{B}$. \label{item:intro-exp-defer-yes}
        \item Otherwise, output $\bot \otimes \sigma_\reg{B}$.
      \end{enumerate}
    \end{minipage}
  }
  \caption{`Real' and `deferred' channels (where $\bot$ is some failure state).}
  \label{fig:intro-experiments}
\end{figure}
\fi

We formalize this idea by introducing two quantum channels (Figure \ref{fig:intro-experiments}): $\mathsf{Real}$, which, when $\mathcal{A} = \mathcal{T}_V \circ \mathcal{T}_P$, corresponds to the honest application of the QNIC; and $\mathsf{Defer}$, which we will use in our attack.
We will show that $\mathsf{Real}$ and $\mathsf{Defer}$ are indistinguishable for all $\mathcal{A}$ when $(\Enc,\Dec)$ is QCA-secure.
Given that, we construct an attack (depicted in \cref{fig:intro-attack}) which forges a proof for a randomly sampled counterexample protocol.
After switching from $\mathsf{Real}$ to $\mathsf{Defer}$ we can replace $\prover$ with its zero knowledge simulator $\Sim$.
Given a statement $x$, our forgery attack $\mathcal{F}$ proceeds as follows:
\begin{enumerate}
  \item Sample a random $k \gets \Gen(1^n)$ and run $\mathcal{T}_P$.
  \item Once $\mathcal{T}_P$ requests an initial prover message, send a ciphertext encrypting half of a maximally entangled state $\Phi^+$.
  \item After $\mathcal{T}_P$ sends a challenge $c$, run $\Sim(x, c)$ to generate $\reg{AZ}$ and respond with $\reg{Z}$ (but hold onto $\reg{A}$).
  \item After $\mathcal{T}_P$ produces a proof $\pi$, provide it to $\mathcal{T}_V$ to get a transcript state on $\reg{A'C'Z'}$.
    Measure if $\reg{A'}$ is a valid ciphertext. If not, then abort.
    Otherwise, replace $\reg{A'}$ with an encryption of $\reg{A}$ output by $\Sim$.
  \item Uncompute $\mathcal{T}_V$ to get a forged proof $\pi'$.
\end{enumerate}

Assuming $\mathsf{Real}$ and $\mathsf{Defer}$ are indistinguishable and that $\Sim$ outputs a transcript state which is indistinguishable from one which is honestly generated, the forged proof $\pi'$ should be indistinguishable from an honestly generated proof for $x$.
Because the QNIC is $(1 - \varepsilon)$-complete, the proof produced by $\mathcal{F}(x)$ for $x \in \lang$ will be accepted by the NIZK verifier with probability $\geq 1 - \varepsilon(n) - \negl(n)$.
In the same way, we know that $\pi' \gets \mathcal{F}(x)$ for $x \notin \lang$ must be accepted with probability $< \delta$ because the QNIC is $\delta$-sound.
Thus combining $\mathcal{F}$ with the NIZK verifier yields an efficient decider for the language $\lang$.

Observe that $\mathcal{F}$ relativizes quite straightforwardly because it makes no assumption on the inner workings of $\mathcal{T}_P$ or $\mathcal{T}_V$.
The only caveat is that $\mathcal{F}$ must be able to uncompute $\mathcal{T}_V$.
Therefore if $\mathcal{T}_V$ has access to some oracle $\oracle$, we require soundness to hold against adversaries with access to $(\oracle,\oracle^\dagger)$.

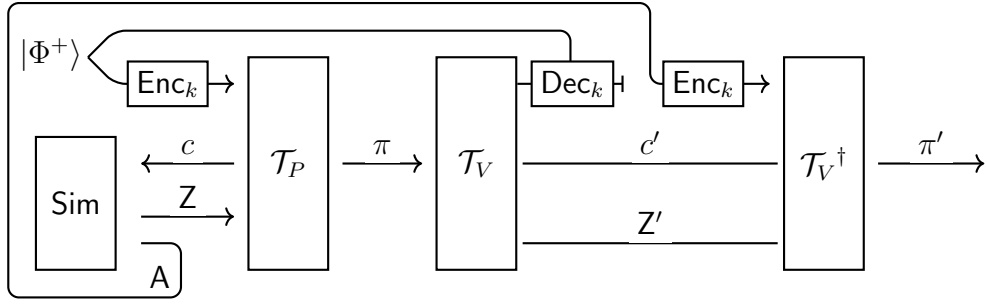
\begin{figure}[t]
  \centering
  \begin{tikzpicture}[
    font=\large,
    box/.style={draw, thick, minimum width=3em, minimum height=8em},
    msg/.style={->, thick, shorten <=0.5em, shorten >=0.5em},
    label/.style={midway, fill=white, inner sep=3pt}
  ]
    \def\padding{0.5em}

    \node[box, minimum height=5em, anchor=south] (P) {$\Sim$};
    \node[box, right=8em of P.south west, anchor=south west] (TP) {$\mathcal{T}_P$};
    \node[box, right=4em of TP] (TV) {$\mathcal{T}_V$};
    \node[box, right=10em of TV] (V) {${\mathcal{T}_V}^\dagger$};

    \node at ($(TP.west)+(-7.5em, 4.1em)$) {$\ket{\Phi^+}$};

    \draw[-, thick, rounded corners] ($(TP.west)+(-6.0em, 4em)$) -- ($(TP.west)+(-5.0em, 5em)$) -- ($(TV.east)+(2em, 5em)$) -- ($(TV.east)+(2em, 3.5em)$);
    \draw[->, thick, rounded corners, shorten >=0.5em] ($(TP.west)+(-6.0em, 4em)$) -- ($(TP.west)+(-5.0em, 3em)$) -- ($(TP.west)+(0, 3em)$);
    \node[draw,thick,fill=white] at ($(TP.west)+(-3em,3em)$) (Enc) {$\Enc_k$};
    \draw[msg] ($(TP.west)+(0, 0mm)$) -- ($(TP.west)+(-4.5em, 0mm)$) node[label, above] {$c$};
    \draw[msg] ($(TP.west)+(-4.5em,-2em)$) -- ($(TP.west)+(0,-2em)$) node[label, above] {$\reg{Z}$};

    \draw[->, thick,rounded corners]
      ($(TP.west)       + (-4.0em,-3em)$) --
      ($(TP.west)       + (-2.5em,-3em)$) --
      ($(TP.south west) + (-2.5em,-1em)$) --
      ($(P.south)       + (0,-1em)$)      --
      ($(P.south west)  + (-1em,-1em)$)   --
      ($(TP.north west) + (-9em,2em)$)    --
      ($(TV.north east) + (+5em,2em)$)    --
      ($(TV.east)       + (+5em,3em)$)   --
      ($(V.west)        + (-0.5em, 3em)$);
    \node at ($(TP.south west) + (-3.3em,-0.24em)$) {$\reg{A}$};
    \node[draw,thick,fill=white] at ($(V.west)+(-3em,3em)$) (Enc) {$\Enc_k$};

    \draw[msg] (TP.east) -- (TV.west) node[label, above] {$\pi$};

    \draw[-|, thick] ($(TV.east)+(0em, 3em)$) -- ($(TV.east)+(4em, 3em)$);
    \node[draw,thick,fill=white,anchor=west] at ($(TV.east)+(0.5em,3em)$) (Dec) {$\Dec_k$};
    \draw[-, thick] ($(TV.east)+(0.2em, 0)$) -- ($(V.west)+(-0.2em, 0)$) node[label, above] {$c'$};
    \draw[-, thick] ($(TV.east)+(0.2em,-3em)$) -- ($(V.west)+(-0.2em,-3em)$) node[label, above] {$\reg{Z'}$};


    \draw[msg] (V.east) -- ($(V.east)+(5em,0)$) node[label, above] {$\pi'$};

  \end{tikzpicture}
  \caption{Attack on Straightline QNIC}
  \label{fig:intro-attack}
\end{figure}

\paragraph{Deferring the Message.}

It remains to be shown that $\mathsf{Real}$ and $\mathsf{Defer}$ are in fact indistinguishable.
At first glance, it might appear that this follows immediately from QCA security:
$\mathsf{Real}$ and $\mathsf{Defer}$ resemble the real and simulated channels $\Gamma$ and $\Lambda$, respectively.
However, there is a crucial difference: $\Gamma$ and $\Lambda$ output a decrypted plaintext, whereas $\mathsf{Real}$ and $\mathsf{Defer}$ output a ciphertext.
It is therefore not clear that a distance guarantee on $\Gamma$ and $\Lambda$ implies a distance guarantee on $\mathsf{Real}$ and $\mathsf{Defer}$.



In order to bridge this gap, we introduce a property for QCA schemes which we call {\em retrospective security}.
An encryption scheme is retrospectively secure if an adversary cannot distinguish between the real and simulated channels, even if both $k$ and $r$ are \emph{revealed} upon a successful decryption. This resolves our issue: a distinguisher for $\mathsf{Real}$ and $\mathsf{Defer}$ can be transformed into a retrospective QCA distinguisher by first reencrypting the plaintext under $k,r$.

Note that we only disclose the secret key if the adversary's ciphertext is valid.
Otherwise, an adversary could simply hold onto its given ciphertext, provide a junk state to $\Dec$, and decrypt the ciphertext itself once given the key.
By returning a valid ciphertext to $\Dec$, the adversary is demonstrating that it has `forgotten' any information it had about the message.
Although retrospective security appears to be stronger than plain QCA security, we prove that all QCA-secure schemes satisfy the property.

We prove this implication using a continuity theorem for Stinespring operators from \cite{KretschmannSW08}.
At a high level, the theorem provides the following guarantee.
Let $\Channel{X,Y}{A}{B}$ be channels and let $\Channel{X'}{A}{BC}$ be any expansion of $\mathcal{X}$ such that $\Tr_\reg{C} \circ \mathcal{X}' = \mathcal{X}$.
If $\| \mathcal{X} - \mathcal{Y} \|_\diamond$ is small, then there exists a corresponding expansion $\Channel{Y'}{A}{BC}$ such that $\| \mathcal{X}' - \mathcal{Y}' \|_\diamond$ is small.
In our setting, we have real and simulated channels $\Gamma$ and $\Lambda$ which we know by the QCA guarantee are negligibly close, and we have an expansion of $\Gamma$, which we denote $\Gamma_\mathrm{R}$, which also outputs $(k,r)$ on a new register $\reg{K}$.
Then by the continuity theorem there exists some expansion $\Lambda_\mathrm{R}$ of $\Lambda$ to include $\reg{K}$ such that $\Gamma_\mathrm{R}$ and $\Lambda_\reg{R}$ are negligibly close as well.
We then show that $\Lambda_\mathrm{R}$ is close to the expansion of $\Lambda$ which outputs $(k,r)$ on $\reg{K}$.

In addition, there is a subtle malleability issue that arises when showing $\mathsf{Real}$ and $\mathsf{Defer}$ are close.
If the attacker in $\mathsf{Real}$ is somehow able to map a ciphertext with randomness $r$ to a valid ciphertext with randomness $r' \neq r$, then the output of $\mathsf{Real}$ and $\mathsf{Defer}$ could be distinguished: $\mathsf{Real}$ will produce a ciphertext encrypted with $r'$, whereas $\mathsf{Defer}$ always produces a ciphertext encrypted with $r$.
This type of attack does not arise in either the original or retrospective QCA security games, since the encryption randomness of the adversary's ciphertext is discarded.
Nonetheless, in \cref{sec:randomness-preservation} we show via an operator-theoretic argument that ciphertexts in a QCA-secure scheme are not malleable in this way.

\begin{remark}
  Readers familiar with the quantum authentication literature may notice that retrospective security closely resembles definitions for key recycling.
  In particular, \cite{DulekS18} defines ``QCA with key recycling'' which is close to our definition but where the simulator outputs a freshly sampled key rather than the key used in the simulation.
  The authors note that this is a strictly stronger notion than QCA (and therefore our retrospective definition) and provide a simple separating construction where an additional key bit is added to the ciphertext.
  Another work \cite{GargYZ16} defines a security notion ``total authentication with key leakage'' which is closest to our retrospective definition.
  In fact, our definition could be understood as a strengthening of their definition from plaintext to ciphertext authentication, if not for our inclusion of the encryption randomness.
\end{remark}

\paragraph{Roadmap.}

The structure for the rest of the paper is as follows.
Notation and definitions are laid out in \cref{sec:prelims}.
In particular, we point readers towards \cref{sec:proof-systems} which describes our unconventional modeling and notational choices for representing $\Xi$-protocols.
In \cref{sec:model}, we formally define QNICs and discuss the black-box access that they are given to the $\Xi$-protocols they are compiling.
\cref{sec:retrospective} proves the indistinguishability of $\mathsf{Real}$ and $\mathsf{Defer}$ by tackling both gaps described above (in \cref{sec:retrospective-security} and \cref{sec:randomness-preservation}).
Finally, the formal statement and proof of \cref{thm:main-result} appear in \cref{sec:impossibility}.

\paragraph{Disclosure on Use of AI.}
The authors used ChatGPT (v5.6 Sol and v6 Astra) as a sounding board, for proofreading, and for some of the technical ideas.
Specifically, ChatGPT supplied some of the ideas used in the proof of \cref{lem:preserve-randomness}, and the full proof strategy for \cref{lem:qca-implies-retrospective}.
The paper itself was written and verified in its entirety by the authors, and they take full responsibility for the correctness and content of the work.


\section{Preliminaries}
\label{sec:prelims}

By $\negl(n)$ we denote any function which is asymptotically smaller than every inverse polynomial --- i.e., that is $o(n^{-c})$ for every $c \in \mathbb{N}$.
Given a set $S$, we write $s \gets S$ to signify that $s$ is uniformly sampled from $S$.
If $\mathbf{D}$ is a distribution, then $x \gets \mathbf{D}$ denotes that $x$ is chosen according to that distribution.

\begin{definition}
  Let $(S, \leq)$ be a partially ordered set and $T \subseteq S$. Then $T$ is {\em downward-closed} if, for all $b \in T$ and $a \leq b$, it holds that $a \in T$.
\end{definition}

\begin{definition}
  Let $X$ be a Hermitian operator such that $X = \sum_{k=1}^m \lambda_k \Pi_k$ for real $\lambda_k$ and non-zero projections $\Pi_k$.
  Then its {\em Jordan-Hahn decomposition} is given by
  \[
    X_{+} = \sum_{k : \lambda_k > 0} \lambda_k \Pi_k,
    \hspace{5mm}\text{and}\hspace{5mm}
    X_{-} = \sum_{k : \lambda_k < 0} -\lambda_k \Pi_k
  \]
  where $X_{+}, X_{-} \geq 0$, $X = X_{+} - X_{-}$, and $X_{+}X_{-} = 0$.
\end{definition}

\begin{fact} \label{fact:trace-pos-eigenvalues}
  Let $X$ be Hermitian. Then $\Tr[X_+] = \frac{1}{2}(\|X\|_1 + \Tr[X])$.
\end{fact}
\begin{proof}
  Because $X$ is Hermitian, $\|X\|_1 = \Tr\sqrt{X^*X} = \Tr[X_{+}] + \Tr[X_{-}]$ whereas $\Tr[X] = \Tr[X_{+}] - \Tr[X_{-}]$.
  Adding these equations gives $2\Tr[X_{+}] = \|X\|_1 + \Tr[X]$.
\end{proof}

\subsection{Quantum Information}

A pure quantum state is a vector $\ket{\psi}$ in a complex Hilbert space $\mathcal{H}$ such that $\|\ket{\psi}\|_2 = 1$.
In this work, $\mathcal{H}$ is finite-dimensional and described as a register --- e.g., $\mathcal{H} = \reg{A} \otimes \reg{B}$ where $\reg{A}$ and $\reg{B}$ are subspaces.
Let $\LinearOps{A}$ denote the space of linear operators on $\reg{A}$ and $\CP{A}{B}$ denote the cone of completely-positive maps from $\LinearOps{A}$ to $\LinearOps{B}$.
A general (potentially mixed) quantum state is a density matrix $\rho \in \States{A}$ which is positive semidefinite and has unit trace.
The density matrix corresponding to the pure state $\ket{\psi}$ is $\ketbra{\psi}$ which we occasionally shorten to $\psi$.
Subnormalized quantum states are positive semidefinite operators $\rho$ where $0 \leq \Tr[\rho] \leq 1$.
Given two $n$-qubit registers $\reg{A}$ and $\reg{B}$, we denote by $\ket{\Phi^+}_\reg{AB}$ the maximally entangled state $\ket{\Phi^+} = \frac{1}{\sqrt{2^n}}\sum_{x \in \bits^n} \ket{x}_\reg{A} \ket{x}_\reg{B}$.

Binary measurement of a quantum state can be represented by an operator $0 \leq M \leq \Id$, corresponding to the POVM $\{M, \Id - M\}$.
On an input state $\rho$, the $M$ outcome occurs with probability $\Tr[M\rho]$ and leaves the residual state $\rho' = \sqrt{M}\rho\sqrt{M}/\Tr[M\rho]$.
If $\Tr[M\rho] \geq 1 - \varepsilon$, then $\| \rho - \rho' \|_1 \leq 2\sqrt{\varepsilon}$.
We will also use the following bound related to alternating projections.

\begin{lemma} \label{lem:alternating-projections}
  Let $\Pi_1, \Pi_2$ be orthogonal projectors, $0 \leq M \leq \Id$ commute with $\Pi_1$, and $\rho$ be a density matrix with $\Pi_2\rho\Pi_2 = \rho$.
  Then
  \[
    \Tr[M\Pi_1\rho\Pi_1] - \Tr[M\Pi_2\Pi_1\rho\Pi_1\Pi_2]
    \leq \frac{2}{3\sqrt{3}}.
  \]
\end{lemma}
\begin{proof}
  First, for any operators $0 \leq B \leq A \leq \Id$, completing a square gives
  \ifllncs
    \begin{align*}
    B - ABA
    &= \frac{3}{2}\left[\left(\Id - \frac{A}{\sqrt{3}}\right)B\left(\Id - \frac{A}{\sqrt{3}}\right) - \left(A - \frac{\Id}{\sqrt{3}}\right)B\left(A - \frac{\Id}{\sqrt{3}}\right)\right] \\
    &\leq \frac{3}{2}A\left(\Id - \frac{A}{\sqrt{3}}\right)^2 \\
    &\leq \frac{2}{3\sqrt{3}}\Id
  \end{align*}
  \else
  \[
    B - ABA
    = \frac{3}{2}\left[\left(\Id - \frac{A}{\sqrt{3}}\right)B\left(\Id - \frac{A}{\sqrt{3}}\right) - \left(A - \frac{\Id}{\sqrt{3}}\right)B\left(A - \frac{\Id}{\sqrt{3}}\right)\right]
    \leq \frac{3}{2}A\left(\Id - \frac{A}{\sqrt{3}}\right)^2
    \leq \frac{2}{3\sqrt{3}}\Id.
  \]
  \fi
  where the final inequality follows by maximizing $\frac{3}{2} t (1-t / \sqrt{3})^2$ over $0 \leq t \leq 1$, whose maximum occurs at $t = 1/\sqrt{3}$.
  Apply this bound with $A = \Pi_2\Pi_1\Pi_2$ and $B = \Pi_2\Pi_1 M\Pi_1\Pi_2$, which satisfy $0 \leq B \leq A \leq \Id$.
  Since $M$ commutes with $\Pi_1$,
  \[
    \Pi_2 M\Pi_2
    = B + \Pi_2(\Id-\Pi_1)M(\Id-\Pi_1)\Pi_2
    \geq B.
  \]
  Using $\Pi_2\rho\Pi_2 = \rho$ and cyclicity of the trace, we conclude
  \ifllncs
  \begin{align*}
    \Tr[M\Pi_1\rho\Pi_1] - \Tr[M\Pi_2\Pi_1\rho\Pi_1\Pi_2]
    &= \Tr[(B-A\Pi_2 M\Pi_2 A)\rho] \\
    &\leq \Tr[(B-ABA)\rho] \\
    &\leq \frac{2}{3\sqrt{3}}.
  \end{align*}
  \else
  \[
    \Tr[M\Pi_1\rho\Pi_1] - \Tr[M\Pi_2\Pi_1\rho\Pi_1\Pi_2] = \Tr[(B-A\Pi_2 M\Pi_2 A)\rho] \leq \Tr[(B-ABA)\rho] \leq \frac{2}{3\sqrt{3}}. \qedhere
  \]
  \fi
\end{proof}

Pure quantum states evolve unitarily.
General evolution of a quantum state can be represented through completely-positive, trace-preserving maps (also called channels).
We describe a channel $\mathcal{N}$ as `efficient' if there exists a $\poly$-sized quantum circuit (potentially receiving a quantum advice state) which can realize $\mathcal{N}$.
If a linear map is completely-positive and trace-decreasing, then we say it is a subchannel.
A linear map which acts on maps is called a supermap.
In this work we are only concerned with ``physically realizable'' supermaps which admit a realization $\Theta(\mathcal{N}) = \mathcal{Y} \circ (\mathcal{N} \otimes \Id_\reg{E}) \circ \mathcal{X}$ where $\Channel{N}{A}{B}$, $\Channel{X}{C}{AE}$, and $\Channel{Y}{BE}{D}$ are (sub)channels.
In this work, we overload $\Id$ to denote the identity operator, the channel which acts trivially, and the supermap which acts trivially.

Given a linear map $\Map{N}{A}{B}$, its diamond norm is
\ifllncs
\[
  \| \mathcal{N} \|_\diamond = \sup_{0 \neq X \in \LinearOps{AE}} \|(\mathcal{N} \otimes \Id_\reg{E})(X)\|_1/\|X\|_1
\]
\else
  $\| \mathcal{N} \|_\diamond = \sup_{0 \neq X \in \LinearOps{AE}} \|(\mathcal{N} \otimes \Id_\reg{E})(X)\|_1/\|X\|_1$
\fi
where it suffices to take $\reg{E} \cong \reg{A}$.
If $\mathcal{N}$ is Hermiticity-preserving, the supremum may equivalently be taken over density matrices on $\reg{AE}$.
If $\mathcal{N}$ is completely positive, then $\|\mathcal{N}\|_\diamond = \|\mathcal{N}^\dagger(\Id)\|_\infty$.
The diamond distance induced by that norm obeys the data processing inequality --- i.e., it contracts under composition with channels (e.g., $\|\mathcal{N} \circ (\mathcal{X} - \mathcal{Y})\|_\diamond \leq \| \mathcal{X} - \mathcal{Y} \|_\diamond$).
More generally, the norm contracts under (physically realizable) supermaps.

\begin{fact} \label{fact:diamond-norm-partial-trace}
  Let $\Map{N}{A}{BC}$ be completely positive. Then $\|\mathcal{N}\|_\diamond = \|\mathrm{Tr}_\reg{C} \circ \mathcal{N}\|_\diamond$.
\end{fact}
\begin{proof}
  Since $\mathcal{N}$ is completely positive, $\|\mathcal{N}\|_\diamond = \|\mathcal{N}^\dagger(\Id_\reg{BC})\|_\infty$. Also, $(\mathrm{Tr}_\reg{C} \circ \mathcal{N})^\dagger(\Id_\reg{B}) = \mathcal{N}^\dagger(\Id_\reg{B} \otimes \Id_\reg{C})$.
\end{proof}

For all completely-positive maps $\Channel{N}{A}{B}$, there exists a linear operator $\Map*{V}{A}{BE}$ and an additional register $\reg{E}$ such that $\mathcal{N}(\rho) = \Tr_\reg{E}[V \rho V^\dagger]$.
If $\mathcal{N}$ is a channel, then $V^\dagger V = \Id$ and $V$ is an isometry.
If $\mathcal{N}$ is a subchannel, then $V^\dagger V \leq \Id$ and $V$ is a contraction.
The pair $(V, \reg{E})$ is called a Stinespring dilation of $\mathcal{N}$.
\cref{sec:retrospective} relies on the following two results on Stinespring dilations.

\begin{theorem}[\cite{KretschmannSW08}] \label{thm:ksw}
  Let $\mathcal{N}_1, \mathcal{N}_2 : \reg{A} \to \reg{B}$ be completely-positive, trace-non{\ifllncs-\fi}increasing maps with Stinespring operators $V_1, V_2 : \reg{A} \mapsto \reg{BE}$ with a common auxiliary register $\reg{E}$.
  Then
  \[
    \inf_U \|(\Id_\reg{B} \otimes U)V_1 - V_2\|_\infty^2 \leq \|\mathcal{N}_1 - \mathcal{N}_2\|_\diamond \leq 2 \inf_U \|(\Id_\reg{B} \otimes U)V_1 - V_2\|_\infty
  \]
  where the minimization is with respect to all unitaries $U$ on $\reg{E}$.
\end{theorem}

\begin{theorem}[Radon-Nikodym Theorem for Completely-Positive Maps \cite{Belavkin1986radon,Raginsky2003radon}] \label{thm:radon-nikodym}
  Let $\mathcal{N}_1, \mathcal{N}_2 : \reg{A} \to \reg{B}$ be completely-positive maps and let $V : \reg{A} \to \reg{BE}$ be an operator such that $(\reg{E}, V)$ is a Stinespring dilation of $\mathcal{N}_2$.
  Then $\mathcal{N}_1 \leq \mathcal{N}_2$ if and only if there exists an operator $0 \leq Q \leq \Id$ on $\reg{E}$ such that
  \[
    \mathcal{N}_1(X) = \Tr_\reg{E}\left[(\Id_\reg{B} \otimes Q^{\frac{1}{2}}) V X V^\dagger (\Id_\reg{B} \otimes Q^{\frac{1}{2}})\right]
  \]
\end{theorem}

In this work we will be working with both unitary and channel oracles.
In both cases, we denote by $\mathcal{A}^{\oracle}$ the algorithm $\mathcal{A}$ with access to the oracle $\oracle$.
If $\oracle$ is unitary, $\mathcal{A}$ can be understood as having access to $\oracle$, $\oracle^\dagger$, and $\oracle^*$ as well as their controlled versions.
If $\oracle$ is a channel, then $\mathcal{A}$ has access only to the channel.
In \cref{thm:main-result}, by efficiently simulatable we mean the following.

\begin{definition}\label{def:oracle-sim}
  An oracle distribution $\oracledist$ is {\em efficiently simulatable} if there exists a QPT (stateful) algorithm
  $\Sim$ such that, for all QPT $\mathcal{A}$,
  \[
    \left|\Pr_{\oracle \gets \oracledist}[\mathcal{A}^{\oracle}(1^n) = 1] -
    \Pr[\mathcal{A}^{\Sim}(1^n) = 1]\right| \leq \negl(n)
  \]
\end{definition}

\begin{theorem}[\cite{Zhandry18}]
  The QROM distribution is efficiently simulatable.
\end{theorem}

\begin{theorem}[\cite{GrinkoY25,MaH25,FoxmanLMNW26}]
  The QHROM distribution is efficiently simulatable.
\end{theorem}

\subsection{Quantum Commitments}

In order to allow for the existence of post-quantum zero knowledge, we must assume that quantum commitment schemes exist \cite{BrakerskiCQ23}.
Because our result is concerned with zero knowledge in a relativized model, we define quantum commitments with respect to an oracle distribution.

\begin{definition}
  Let $\Com = \{Q_0(\lambda), Q_1(\lambda)\}_{\lambda\in\mathbb{N}}$ be an ensemble of uniformly generated quantum oracle circuits acting on registers $\reg{CD}$\footnote{
    The size of $\reg{CD}$ will depend on the security parameter.
  } (where we will drop the security parameter $\lambda$ to simplify notation).
  Then $\Com$ is a quantum bit commitment scheme with respect to an oracle distribution $\oracledist$ if the following hold.
  \begin{itemize}
    \item {\bf Computational Hiding.} The reduced state on $\reg{C}$ is computationally indistinguishable between $Q_0^\oracle\ket{0}_\reg{CD}$ and $Q_1^\oracle\ket{0}_\reg{CD}$ in expectation over $\oracledist$.
    \item {\bf (Honest) Statistical Binding.} For all auxiliary states $\ket{\psi}$ on $\reg{E}$ and unbounded quantum oracle circuits $U$ acting on $\reg{DE}$, we have that
      \[
        \E_{\oracle \gets \oracledist } \left\| \left(Q_1^{\oracle}\ketbra{0}(Q_1^{\oracle})^\dagger \otimes \Id_\reg{E}\right) (\Id_\reg{C} \otimes U^{\oracle}_\reg{DE}) (Q_0^{\oracle}\ket{0}_\reg{CD}\ket{\psi}_\reg{E}) \right\| \leq \negl(n)
      \]
  \end{itemize}
\end{definition}

\subsection{Proof Systems}
\label{sec:proof-systems}

As it will be convenient in later sections, we use an unconventional but equivalent model to describe $\Xi$-protocols (also known as quantum $\Sigma$-protocols).
Rather than defining the prover as a stateful algorithm which is invoked twice in one interaction, we will split the prover into an initial message $\xi^x$ (which depends on the statement $x$ being proved) and a quantum channel $\Channel{\prover}{CS}{Z}$ which responds to the verifier's challenge.
The state $\xi^x$ includes a message register $\reg{A}$ and an internal register $\reg{S}$ which is used by $\prover$ to formulate its final message on $\reg{Z}$.
This is equivalent to the stateful prover model because we can simply define $\xi^x$ to be the output of $\prover(x)$ along with the prover's residual internal state.
Our choice makes clear what is meant by making a single query to the prover.

Although the challenge and verifier output are classical, we can model the verifier as a quantum channel $\Channel{\verify}{XACZ}{B}$ where $\reg{B}$ is a single qubit.
It can be assumed that $\verify$ measures $\reg{C}$ in the computational basis immediately, and likewise $\verify$'s acceptance can be determined by measuring $\reg{B}$ in the same basis.
Given these modeling choices, the protocol interaction can be described by the circuit in \cref{fig:xi-protocol}.

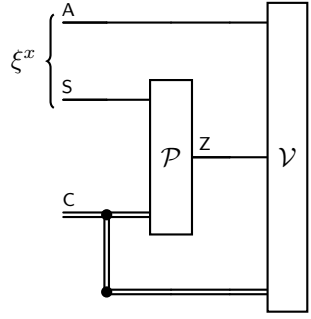
\begin{figure}
  \centering
  \begin{quantikz}
    \lstick[2]{$\xi^x$} & \wire[l][1]["\reg{A}"{above,anchor=south west,pos=1.19}]{q} & & & \gate[5]{\verify} \\
    & \wire[l][1]["\reg{S}"{above,anchor=south west,pos=1.19}]{q} & \gate[3]{\prover} & \setwiretype{n} & \\
    \setwiretype{n} & & \wire[r][1]["\reg{Z}"{above,anchor=south west,pos=0}]{q} & \setwiretype{q} & \\
    \setwiretype{n} & \ctrl[vertical wire=c]{1} \wire[l][1]["\reg{C}"{above,anchor=south west,pos=1.19}]{c}  & \wire{c} & \setwiretype{n} & \\
    \setwiretype{n} & \ctrl{0} & \setwiretype{c} & &
  \end{quantikz}
  \caption{Circuit Representation for an Interactive $\Xi$-Protocol}
  \label{fig:xi-protocol}
\end{figure}

\begin{definition}[$\Xi$-Protocols]
  Let $\lang \subseteq \bits^n$ be a language, $\{\xi^x\}_{x\in\lang} \subseteq \States{AS}$ be a family of quantum states, $\Channel{\prover}{CS}{Z}$ be a quantum channel, $\Channel{\verify}{XACZ}{B}$ be an efficient quantum algorithm and $\pi_{x,c} = (\Id_\reg{A} \otimes \prover)(\xi^x \otimes \ketbra{c})$ be a proof state.
  Then $(\{\xi^x\}, \prover, \verify)$ is a {\em $\Xi$-protocol for $\lang$} if the following holds.
  \begin{itemize}
    \item {\bf Completeness.}
      $\forall\, x \in \lang$,
      \[
        \E_{c \gets \chlspace} \ev{\verify(\ketbra{x,c} \otimes \pi_{x,c})}{1} \geq 1 - \negl(n)
      \]
  \item {\bf Soundness.}
    $\forall$ $x \notin \lang$, $\rho \in \States{AS}$, and $\Channel{\prover^*}{CS}{Z}$,
    \[
      \E_{c \gets \chlspace} \ev{\verify(\ketbra{x,c} \otimes \pi^*_{x,c})}{1} < \negl(n)
    \]
    where $\pi^*_{x,c} = (\Id_\reg{A} \otimes \prover^*)(\rho \otimes \ketbra{c})$.

  \item {\bf Special Honest Verifier Zero Knowledge.}
    $\exists$ an efficient algorithm $\Sim$ such that $\forall\, x \in \lang, c \in \chlspace$,
    \[
      \left\{ x, c, \pi \mid \pi \gets \Sim(x, c) \right\} \eqc
      \left\{ x, c, \pi_{x,c} \right\}
    \]
  \end{itemize}
\end{definition}

\begin{remark}
  In general, $\Xi$-protocols are not required to be overwhelmingly complete nor negligibly sound, but we have made this choice for convenience.
  As we are working towards an impossibility, this will only strengthen our result.
  That is to say, if straightline QNICs cannot be universal for this restricted class of $\Xi$-protocols, then they cannot be universal for a more general definition either.
\end{remark}

\subsection{Quantum Encryption Schemes}

\begin{definition}[SKQES]
  Let $\mathsf{QES} = (\Gen, \Enc, \Dec)$ where $\Gen : \mathbb{N} \to \keyspace$, $\Enc : \keyspace \times \States{M} \to \States{C}$, and $\Dec : \keyspace \times \States{C} \to \States{M} \oplus \bot$ are efficient quantum algorithms.
  Then $\mathsf{QES}$ is a {\em symmetric-key quantum encryption scheme} if for all $k \in \keyspace$
  \[
    \|\Dec_k \circ \Enc_k - \Id_\reg{M} \oplus \mathbf{0}_\bot\|_\diamond \leq \negl(n)
  \]
\end{definition}

\begin{lemma}[\cite{AlagicGM18}, Lemma 1] \label{lem:canonical}
  Let $\mathsf{QES} = (\Gen, \Enc, \Dec)$ be a SKQES. Then there exists a family of unitaries $U_k$, a probability distribution $p_k$, a family of quantum states $\ket{\psi^{(k,r)}}_\reg{T}$, and a family of channels $\Channel{\mathcal{R}_k}{C}{M}$ such that
  \begin{align*}
    \Enc_{k}(\rho_\reg{M}) &= U_k\left(\rho \otimes \sum_{r} p_k(r) \ketbra{\psi^{(k,r)}}\right)U_k^\dagger \\
    \Dec_k(\varrho_{\reg{C}}) &= \Tr_{\reg{T}}\left[U_k^\dagger \Pi_k \varrho \Pi_k U_k \right] + \mathcal{R}_k\left((\Id - \Pi_{k})\varrho(\Id - \Pi_{k})\right)
  \end{align*}
  where $\Pi_k = \sum_{r} \Pi_{k,r}$ and $\Pi_{k,r} = U_k(\Id_\reg{M} \otimes \ketbra{\psi^{(k,r)}})U_k^\dagger$.
\end{lemma}

As it will be useful later, we will additionally denote by $\Enc_{k;r}$ and $\Dec_{k;r}$ the channels
\begin{align*}
  \Enc_{k;r}(\rho_{\reg{M}}) &= U_k(\rho \otimes \ketbra{\psi^{(k,r)}})U_k^\dagger \\
  \Dec_{k;r}(\varrho_{\reg{C}}) &= \Tr_{\reg{T}}\left[U_k^\dagger \Pi_{k,r} \varrho \Pi_{k,r} U_k \right] + \mathcal{R}_k\left((\Id - \Pi_{k,r})\varrho(\Id - \Pi_{k,r})\right)
\end{align*}

\begin{definition}[Quantum Ciphertext Authenticating (QCA)]\label{def:qca}
  Let $\mathsf{QES} = (\Gen, \Enc, \Dec)$ be a SKQES.
  Then $\mathsf{QES}$ is {\em QCA-secure} if, for all quantum channels $\Channel{A}{CR}$, there exists a completely-positive map $\Channel{\Lambda^{\mathrm{rej}}}{R}{R}$ such that $\| \Gamma - \Lambda \|_\diamond \leq \negl(n)$ and $\Lambda$ is trace-preserving where
  \[
    \Gamma = \E_k\left[\Dec_k \circ \mathcal{A} \circ \Enc_k\right] \eqand{2em} \Lambda = \Id_\reg{M} \otimes \E_{k,r}[\Lambda^{\mathrm{acc}}_{k,r}] + \ketbra{\bot}_\reg{M} \otimes \Lambda^{\mathrm{rej}}
  \]
  are the `real' and `simulated' effective attack maps, respectively, given that
  \ifllncs
  \begin{align*}
    \Lambda^{\mathrm{acc}}_{k,r}(X_\reg{R}) = \Tr_\reg{LMT}[\widetilde{\Pi}_{k,r}\left(\mathcal{A} \circ \Enc_{k;r}(\Phi^+_\reg{LM} \otimes X)\right)] \\
    \widetilde{\Pi}_{k,r} = U_k(\ketbra{\Phi^+}_\reg{LM} \otimes \ketbra{\psi^{(k,r)}}_\reg{T})U_k^\dagger
  \end{align*}
  \else
  \[
    \Lambda^{\mathrm{acc}}_{k,r}(X_\reg{R}) = \Tr_\reg{LMT}[\widetilde{\Pi}_{k,r}\left(\mathcal{A} \circ \Enc_{k;r}(\Phi^+_\reg{LM} \otimes X)\right)]
    \eqand{1em}
    \widetilde{\Pi}_{k,r} = U_k(\ketbra{\Phi^+}_\reg{LM} \otimes \ketbra{\psi^{(k,r)}}_\reg{T})U_k^\dagger.
  \]
  \fi
\end{definition}

\begin{fact} \label{fact:dec-before-after-proj}
  Let $\mathsf{QES} = (\Gen, \Enc, \Dec)$ be a SKQES.
  Then
  \ifllncs
  \begin{align*}
    \Pi^\bot \Dec_k(\rho) \Pi^\bot = \Dec_k((\Id - \Pi_k)\rho(\Id - \Pi_k)) \\
    (\Id - \Pi^\bot)\Dec_k(\rho)(\Id - \Pi^\bot) = \Dec_k(\Pi_k\rho\Pi_k)
  \end{align*}
  \else
  \[
    \Pi^\bot \Dec_k(\rho) \Pi^\bot = \Dec_k((\Id - \Pi_k)\rho(\Id - \Pi_k))
    \hspace{1.5em}\text{ and }\hspace{1.5em}
    (\Id - \Pi^\bot)\Dec_k(\rho)(\Id - \Pi^\bot) = \Dec_k(\Pi_k\rho\Pi_k).
  \]
  \fi
\end{fact}


\section{Modelling Quantum Non-Interactivity Compilers (QNICs)}
\label{sec:model}

Our goal in this section is to formalize a broad class of Fiat–Shamir–like transformations which remove interaction from public-coin proof systems.
Because we are working towards an impossibility result, we will narrow our focus to compilers for three-message $\Xi$-protocols.
As our main result applies to compilers that are universal for $\Xi$-protocols, it will also cover any expansion of universality --- e.g., to protocols with a polynomial number of rounds.

Recall from \cref{sec:proof-systems} that a $\Xi$-protocol can be defined by a family of quantum states $\{\xi^x\}_{x \in \lang}$, a quantum channel $\prover$, and an efficient quantum algorithm $\verify$.
In the most general sense, a quantum non-interactivity compiler (QNIC) should be a set of algorithms which, when given black-box access to $(\{\xi^x\}, \prover, \verify)$, can construct a non-interactive argument system.
This intuition leads to our first definition.

\begin{definition}[Quantum Non-Interactivity Compiler (QNIC)]
  Let $\mathcal{T} = (\mathcal{T}_P, \mathcal{T}_V)$ be a pair of quantum poly-time oracle algorithms and $\oracledist$ be a distribution over oracles.
  Then $\mathcal{T}$ is an {\em $(\varepsilon,\delta)$-QNIC relative to $\oracledist$} if the following properties hold for all languages $\lang$ and $\Xi$-protocols $(\{\xi^x\}_{x\in\lang},\prover, \verify)$.
  \begin{itemize}
    \item {\bf $(1 - \varepsilon)$-Completeness.}
      $\forall\, x \in \lang$,
      \[
        \Pr_{O \gets \mathbf{O}}\left[\mathcal{T}_V^{\oracle,\verify} \circ \mathcal{T}_P^{\oracle,\prover}(\xi^x) = 1\right] \geq 1 - \varepsilon
      \]
    \item {\bf $\delta$-Soundness.}
      $\forall\, x \notin \lang$ and QPT $\mathcal{A}$,
      \[
        \Pr_{O \gets \mathbf{O}}\left[\mathcal{T}_V^{\oracle,\verify}(\rho) = 1 \;\middle|\; \rho \gets \mathcal{A}^{\oracle}(x)\right] < \delta
      \]
  \end{itemize}
\end{definition}

In this work, we focus on a narrower set of QNICs which are ``straightline'' in that they execute only a single interaction with the prover and verifier.
This does not restrict access to the ``global'' oracle $\oracle$.

\begin{definition}[Straightline QNIC]
  Let $\mathcal{T} = (\mathcal{T}_P, \mathcal{T}_V)$ be a QNIC.
  Then $\mathcal{T}$ is {\em straightline} if $\mathcal{T}_P$ and $\mathcal{T}_V$ make only a single query to their respective $\prover$ and $\verify$ oracles.
\end{definition}

Because $\mathcal{T}_P$ and $\mathcal{T}_V$ make only one query each, they can be decomposed into two channels --- $(\mathcal{T}_{P,1},\mathcal{T}_{P,2})$ and $(\mathcal{T}_{V,1}, \mathcal{T}_{V,2})$, respectively --- such that (e.g.) $\mathcal{T}_P^{\prover} = \mathcal{T}_{P,2}(\Id \otimes \prover)\mathcal{T}_{P,1}$.
In this way, an honestly compiled NIZK argument can be represented as the circuit in \cref{fig:compiled-nizk} with registers labelled.

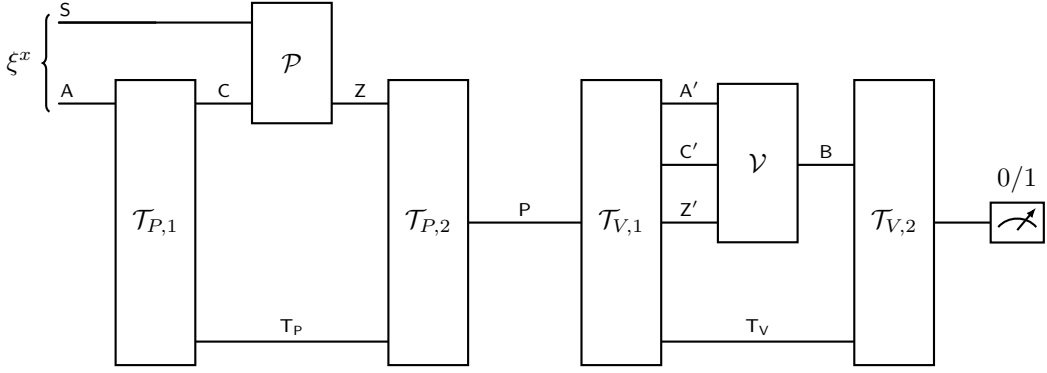
\begin{figure}
  \ifllncs
  \scalebox{0.9}{
  \fi
  \centering
  \begin{quantikz}[wire types={q, n, n, n, n, n}, column sep=7.5mm]
    \lstick[2]{$\xi^x$} & \wire[l][1]["\reg{S}"{above,anchor=south west,pos=1.06}]{q} & \gate[2][3em]{\prover} & \setwiretype{n} \\
    & \gate[5][3em]{\mathcal{T}_{P,1}} \wire[l][1]["\reg{A}"{above,anchor=south west,pos=1.09}]{q} & \wire[l][1]["\reg{C}"{above,pos=0.5}]{q} & \gate[5][3em]{\mathcal{T}_{P,2}} \wire[l][1]["\reg{Z}"{above,pos=0.5}]{q} &
    \setwiretype{n} & \gate[5][3em]{\mathcal{T}_{V,1}} \wire[r][1]["\reg{A'}"{above,pos=0.5}]{q} & \gate[3][3em]{\verify} & \gate[5][3em]{\mathcal{T}_{V,2}} \\
    & & & &
    & \wire[r][1]["\reg{C'}"{above,pos=0.5}]{q} & \wire[r][1]["\reg{B}"{above,pos=0.5}]{q} & & & \\
    & & & \wire[r][2]["\reg{P}"{above,pos=0.5}]{q} &
    & \wire[r][1]["\reg{Z'}"{above,pos=0.5}]{q} & & \wire[r][1]{q} & \meter{0/1} \\
    & & & &
    & & & & \\
    & \wire[r][2]["\reg{T_P}"{above,pos=0.5}]{q} & & &
    & \wire[r][2]["\reg{T_V}"{above,pos=0.5}]{q} & & &
  \end{quantikz}
  \ifllncs
  }
  \fi
  \caption{Circuit Representation of a Non-Interactive Argument Compiled via a Straightline QNIC}
  \label{fig:compiled-nizk}
\end{figure}

\begin{remark}
  It is worth taking time to consider the quantum oracle access that we allow for our QNIC algorithms.
  The choice for oracle access in the quantum setting is typically unitary implementations.
  In our context this would allow (e.g.) $\mathcal{T}_P$ to both compute and uncompute $\prover$.
  In the zero knowledge setting, allowing for inverse access to the prover may be too powerful as it could fundamentally break the zero knowledge of the compiled argument.
  Indeed a trivial compiler could simply extract a witness from the prover and use it directly as a non-interactive argument.

  Because our results are constrained to the single-query setting, we can essentially defer these modeling questions.
  Under the typical black-box assumption that the oracle's private ancilla are accessible only through queries to the oracle, a single unitary query is exactly equivalent to a single corresponding channel query.
  Also, with only a single query, inverse access is not meaningful: the unitary implementation can be chosen such that an inverse query on fresh ancilla will act as the identity on the visible registers.
  Thus for straightline QNICs it suffices to consider channel implementations for $\prover$ and $\verify$.
\end{remark}
\begin{remark}
  Note that we do require that the global oracle $\oracle$ is unitary and thus allows for inverse queries.
  This is because the decision algorithm we construct for \cref{thm:main-result} needs to uncompute $\mathcal{T}_{V,1}^\oracle$.
  Otherwise, we allow arbitrary\footnote{
    As suggested in \cite{Zha25}, this can include conjugate and/or controlled queries as well.
  } access to a polynomial number of queries to $\oracle$.
  Thankfully this is not a particularly onerous requirement and it is satisfied by most oracle models.
  In particular, this captures any quantum-accessible classical function.
\end{remark}

\section{Ciphertext Authentication \& Retrospective Security} \label{sec:retrospective}

In this section, we show that a quantum authentication scheme can be employed to defer the choice of message being encrypted until after a man-in-the-middle acts.
This result is used in \cref{sec:impossibility} to reorder the messages in the straightline interaction between the prover and the QNIC algorithms.

Let $(\Gen, \Enc, \Dec)$ be a quantum encryption scheme.
Recall from \cref{lem:canonical} that there exists a probability distribution $p_k$ over that scheme's encryption randomness and projectors $\Pi_k$ and $\widetilde{\Pi}_{k,r}$ which project onto valid real and simulated ciphertexts, respectively.
Then we can formalize $\mathsf{Real}$ and $\mathsf{Defer}$ from \cref{sec:tech-overview} into subchannels $\Exp_\mathsf{real}$ and $\Exp_\mathsf{defer}$ (see \cref{fig:experiments}) which are parameterized by a man-in-the-middle channel $\Channel{A}{CR}$.
Notice that any action on the message register $\reg{M'}$ commutes with all steps prior to Step \ref{item:exp-defer-step} in $\Exp_\mathsf{defer}$.
This is exactly the deferral property we want, as it will allow us to change the message in response to $\mathcal{A}$'s actions.

\begin{figure}
  \centering
  \newlength{\boxheight}
  \ifllncs \setlength{\boxheight}{11.5em} \else \setlength{\boxheight}{8.7em} \fi
  \fbox{%
    \begin{minipage}[t][\boxheight]{0.4\linewidth}
      $\Exp^{\mathcal{A}}_\mathsf{real}(\rho_\reg{MR})$
      \begin{enumerate}
        \item Sample $k \gets \Gen(1^n)$ {\ifllncs \\ \fi}and $r \gets p_k$.
        \item Prepare $\mathcal{A} \circ \Enc_{k;r}(\rho_\reg{MR})$.
        \item Measure $\reg{C}$ with $\{\Pi_k, \Id - \Pi_k\}$.
        \item Post-select onto $\Pi_k$.
        \item Output $\reg{CR}$ and $k$.
      \end{enumerate}
    \end{minipage}
  }
  \fbox{%
    \begin{minipage}[t][\boxheight]{0.4\linewidth}
      $\Exp^{\mathcal{A}}_\mathsf{defer}(\rho_\reg{M'R})$
      \begin{enumerate}
        \item Sample $k \gets \Gen(1^n)$ {\ifllncs \\ \fi}and $r \gets p_k$.
        \item Prepare $\mathcal{A} \circ \Enc_{k;r}(\Phi^+_\reg{LM} \otimes \rho)$.
        \item Measure $\reg{LC}$ with {\ifllncs \\ \fi}$\{\widetilde{\Pi}_{k,r}, \Id - \widetilde{\Pi}_{k,r}\}$.
        \item Post-select onto $\widetilde{\Pi}_{k,r}$. \label{item:exp-defer-pre-enc}
        \item Apply $\Enc_{k;r}$ to $\reg{M'}$ to get $\reg{C'}$. \label{item:exp-defer-step}
        \item Output $\reg{C'R}$ and $k$.
      \end{enumerate}
    \end{minipage}
  }
  \caption{`Real' and `Deferring' Experiments}
  \label{fig:experiments}
\end{figure}

With these formal subchannels defined, the rest of this section is dedicated to proving the following lemma.

\begin{lemma} \label{lem:exps-close}
  Let $(\Gen, \Enc, \Dec)$ be a SKQES with QCA security and $\Channel{A}{CR}$ be a quantum channel.
  Then $\| \mathsf{Exp}^\mathcal{A}_\mathsf{real} - \mathsf{Exp}^\mathcal{A}_\mathsf{defer} \|_\diamond \leq \negl(n)$.
\end{lemma}

As mentioned in \cref{sec:tech-overview}, this is accomplished through a data processing argument.
That is, we demonstrate $\Exp_\mathsf{real}$ and $\Exp_\mathsf{defer}$ are indistinguishable from the real and simulated effective attack maps from \cref{def:qca} composed with a fixed post-processing step.
Because the post-processing done in both experiments depends on the choice of classical randomness $(k, r)$, we need retrospective security --- i.e., a property of QCA security which allows $(k,r)$ to be revealed to the attacker upon a successful decryption.
Its formal definition is given in \cref{sec:retrospective-security} along with a proof that all QCA secure schemes satisfy the property.

The second catch in our post-processing argument comes from Step \ref{item:exp-defer-step} of $\Exp_\mathsf{defer}$.
Because the re-encryption in that step reuses the original encryption randomness $r$, the output will be a ciphertext entirely in the image $\Pi_{k,r}$ (i.e., the projector onto valid ciphertexts encrypted with key $k$ and randomness $r$).
In $\Exp_\mathsf{real}$, the ciphertext output on $\reg{C}$ will be in the image of $\Pi_k$ but could potentially be partly or mostly outside of $\Pi_{k,r}$ for the original $r$ from Step 1.
In order to resolve this potential vulnerability, we prove in \cref{sec:randomness-preservation} that only negligible weight will be placed in the image of $\Pi_{k} - \Pi_{k,r}$.
Put another way, $\Exp_\mathsf{real}$'s output is indistinguishable from the output of a modified $\Exp_\mathsf{real}$ where $\Pi_{k,r}$ is used instead of $\Pi_k$.
We prove that this `randomness preservation lemma' follows from QCA security by analyzing a ``Bell projection'' supermap which is implicitly used in the QCA simulation.

\subsection{Randomness Preservation}
\label{sec:randomness-preservation}

\begin{lemma} \label{lem:preserve-randomness}
  Let $(\Gen, \Enc, \Dec)$ be a SKQES, $\mathcal{A} : \reg{CR} \to \reg{CR}$ be a channel, $\Gamma$ be the corresponding effective attack map, and $\Gamma^\mathrm{acc}$ be the acceptance branch of that map.
  Then we define a modified attack map
  \[
    \widehat{\Gamma} = \E_{k,r}\left[\Dec_{k;r} \circ \mathcal{A} \circ \Enc_{k;r}\right]
  \]
  where $\Dec_{k;r}(X)$ accepts only ciphertexts which are encrypted with randomness $r$, and a corresponding acceptance branch $\widehat{\Gamma}^{\mathrm{acc}}$.
  If $(\Gen, \Enc, \Dec)$ is QCA secure, then $\|\Gamma^{\mathrm{acc}} - \widehat{\Gamma}^\mathrm{acc}\|_\diamond \leq \negl(n)$.
\end{lemma}

\ifllncs
\begin{proof}
  See \cref{sec:randomness-preservation-proof}.
\end{proof}
\else
In order to prove \cref{lem:preserve-randomness}, we start by proving a general technical lemma about subchannels and their behaviour under supermaps.
\unless\ifllncs
In particular, we are interested in supermaps which are {\em retractions} in that their output always lies within a particular subspace on which the supermap acts trivially.
A retraction is analogous to a projection: it is idempotent, has image $\mathrm{M}$, and fixes every element of $\mathrm{M}$.

\begin{definition}
  Let $\mathrm{M} \subseteq \CP{A}$ be a subset of completely positive maps and $\Map{\Theta}{A}{A}$ be a supermap. Then $\Theta$ is a {\em retraction onto $\mathrm{M}$} if $\Theta(\mathcal{N}) \in \mathrm{M}$ for all $\mathcal{N} \in \CP{A}$ and $\Theta(\mathcal{M}) = \mathcal{M}$ for all $\mathcal{M} \in \mathrm{M}$.
\end{definition}
\fi

\begin{lemma} \label{lem:close-operators}
  Let $\mathcal{X}, \Map{Z}{AB}{CD}$ be completely positive, trace-non{\ifllncs\linebreak[4]-\fi}increasing maps such that $\mathcal{Z} \leq \mathcal{X}$ and $\Map{\Theta}{AB,CD}{AB,CD}$ be a retraction onto a downward-closed subset $\mathrm{M} \subseteq \CP{AB}{CD}$ of completely positive maps.
  If $\|\mathcal{X} - \Theta(\mathcal{Z})\|_\diamond \leq \varepsilon$, then $\|\mathcal{X} - \mathcal{Z} \|_\diamond \leq \varepsilon + 4\sqrt{\varepsilon}$.
\end{lemma}

In plain language, we have a subchannel $\mathcal{X}$, a branch of that subchannel $\mathcal{Z}$, and a supermap $\Theta$ which is a retraction onto a subspace $\mathrm{M}$.
If $\mathcal{X}$ is close to the retraction of $\mathcal{Z}$ onto $\mathrm{M}$, then $\mathcal{Z}$ must be the dominating branch of $\mathcal{X}$.
If $\mathcal{X} \in \mathrm{M}$, then $\mathcal{Z} \in \mathrm{M}$ by downward closure and the result becomes trivial.

Before we prove \cref{lem:close-operators}, we give motivation by showing that it implies \cref{lem:preserve-randomness}.
In that setting, $\Theta$ becomes the ``Bell-projection'' described in \cref{fact:bell-projection} which retracts onto the subspace of maps which act trivially on the first register.
This supermap is implicit in the QCA definition as part of the transformation from the real channel $\Gamma$ to the simulated channel $\Lambda$.
In fact, we have that $\Theta(\widehat{\Gamma}^{\mathrm{acc}}) = \Id_\reg{M} \otimes \E_{k,r}[\Lambda^{\mathrm{acc}}_{k,r}]$.
\ifllncs
We now prove \cref{lem:preserve-randomness} assuming \cref{lem:close-operators} and then prove \cref{lem:close-operators}.
\else
We start with facts about $\Theta$ and its subspace, prove \cref{lem:preserve-randomness}, and then prove \cref{lem:close-operators}.
\begin{fact} \label{fact:bell-projection}
  The ``Bell-projection'' supermap $\Map{\Theta}{AB}{AC}$ given by
  \[
    \Theta(\mathcal{N})(\rho) = \Tr_\reg{ZA'}[(\Phi^+_\reg{ZA'} \otimes \Id_\reg{AC})(\mathcal{N}_{\reg{A'B}\to\reg{A'C}} \otimes \Id_\reg{ZA})(\Phi^+_\reg{ZA'} \otimes \rho_\reg{AB})].
  \]
  is a retraction onto the subspace $\mathrm{M} = \{ \Id_\reg{A} \otimes \mathcal{B} :  \mathcal{B} \in \LinearOps{B}{C} \} \subset \LinearOps{AB}{AC}$.
\end{fact}
\begin{proof}
  Clearly $\Theta(\mathcal{N})$ acts as the identity on $\reg{A}$ by construction for any $\mathcal{N} \in \LinearOps{AB}{AC}$.
  Additionally $\Theta$ acts trivially upon any $\mathcal{N} = \Id_\reg{A} \otimes \mathcal{B}$ for $\mathcal{B} \in \LinearOps{B}{C}$.
  \begin{align*}
    \Theta(\Id_\reg{A} \otimes \mathcal{B})(\rho)
    &= \Tr_\reg{ZA'}[(\Phi^+_\reg{ZA'} \otimes \Id_\reg{AC})(\Id_\reg{ZAA'} \otimes \mathcal{B}_{\reg{B}\to\reg{C}})(\Phi^+_\reg{ZA'} \otimes \rho_\reg{AB})] \\
    &= \Tr_\reg{ZA'}[(\Id_\reg{ZAA'} \otimes \mathcal{B}_{\reg{B}\to\reg{C}})(\Phi^+_\reg{ZA'} \otimes \rho_\reg{AB})] \\
    &= (\Id_\reg{A} \otimes \mathcal{B}_{\reg{B}\to\reg{C}})(\rho_\reg{AB}) \qedhere
  \end{align*}
\end{proof}

\begin{fact} \label{fact:downward-closed}
  {\unless\ifllncs The subspace \fi}$\mathrm{M} = \{ \Id_\reg{A} \otimes \mathcal{B} : \mathcal{B} \in \CP{B}{C} \} \subset \CP{AB}{AC}$ is downward-closed.
\end{fact}
\begin{proof}
  Let $\mathcal{X} = \Id_\reg{A} \otimes \mathcal{B} \in \mathrm{M}$ be a map in the subspace where $\mathcal{B} \in \LinearOps{B}{C}$ and $\mathcal{Z}$ be some map such that $\mathcal{Z} \leq \mathcal{X}$.
  Additionally, let $(V^\mathcal{B}, \mathrm{E})$ be a Stinespring dilation of $\mathcal{B}$ and note that $(\Id \otimes V^{\mathcal{B}}, \reg{E})$ is a Stinespring dilation of $\mathcal{X}$.
  By \cref{thm:radon-nikodym}, there exists $0 \leq Q \leq \Id_\reg{E}$ such that
  \[
    \mathcal{Z}(X) = \Tr_\reg{E}[(\Id_\reg{AC} \otimes Q^{\frac{1}{2}})(\Id_\reg{A} \otimes V^{\mathcal{B}}) X (\Id_\reg{A} \otimes V^{\mathcal{B}})^\dagger(\Id_\reg{AC} \otimes Q^{\frac{1}{2}})]
  \]
  which acts trivially on the $\reg{A}$ register. Therefore $\mathcal{Z} \in \mathrm{M}$ and $\mathrm{M}$ is downward-closed.
\end{proof}
\fi

\begin{proof}[Proof of \cref{lem:preserve-randomness}]
  This follows from \cref{lem:close-operators} where $\mathcal{X} = \Gamma^{\mathrm{acc}}$, $\mathcal{Z} = \widehat{\Gamma}^{\mathrm{acc}}$, and $\Theta$ is the Bell-projection supermap.
  Note that $\Theta(\widehat{\Gamma}^{\mathrm{acc}}) = \Lambda^{\mathrm{acc}}$ and thus we have the starting assumption that $\|\mathcal{X} - \Theta(\mathcal{Z})\|_\diamond \leq \negl(n)$.
  Facts \ref{fact:bell-projection} and \ref{fact:downward-closed} prove that $\Theta$ is a retraction onto a subspace which is downward-closed.
  It remains then to be shown that $\widehat{\Gamma}^{\mathrm{acc}} \leq \Gamma^{\mathrm{acc}}$.

  For brevity, let $\widetilde{X}_{k,r} = \mathcal{A} \circ \Enc_{k;r}(X)$ for a fixed $X$ and a particular $(k,r)$.
  Then we start by writing $\Gamma^{\mathrm{acc}}$ in terms of its constituent components:
  \begin{align*}
    \Gamma^{\mathrm{acc}}(X_\reg{MR})
      &= \E_{k,r}[(\Id - \Pi^\bot)\Dec_k(\widetilde{X}_{k,r})(\Id -\Pi^\bot)] \\
      &= \E_{k,r}\left[\Tr_\reg{T}[U_k^\dagger\Pi_k\widetilde{X}_{k,r}\Pi_kU_k]\right] \\
      \ifllncs
      &= \E_{k,r}\left[\Tr_\reg{T}\left[\left(\sum_s \Id_\reg{M} \otimes \ketbra{\psi^{(k,s)}}\right)U_k^\dagger\widetilde{X}_{k,r}\right.\right. \\
      & \hspace{4cm} \left.\left.U_k\left(\sum_s \Id_\reg{M} \otimes \ketbra{\psi^{(k,s)}}\right)\right]\right] \\
      \else
      &= \E_{k,r}\left[\Tr_\reg{T}\left[\left(\sum_s \Id_\reg{M} \otimes \ketbra{\psi^{(k,s)}}\right)U_k^\dagger\widetilde{X}_{k,r}U_k\left(\sum_s \Id_\reg{M} \otimes \ketbra{\psi^{(k,s)}}\right)\right]\right] \\
      \fi
      &= \E_{k,r}\left[\Tr_\reg{T}\left[\left(\sum_s \Id_\reg{M} \otimes \ketbra{\psi^{(k,s)}}\right)U_k^\dagger\widetilde{X}_{k,r}U_k\right]\right]
  \end{align*}
  where the final equality uses the cyclic property of the partial trace (and that $\{\ket{\psi^{(k,r)}}\}$ are orthogonal).
  Through the same process, we find that
  \[
    \widehat{\Gamma}^{\mathrm{acc}}(X_\reg{MR}) = \E_{k,r}\left[\Tr_\reg{T}\left[(\Id \otimes \ketbra{\psi^{(k,r)}})U_k^\dagger\widetilde{X}_{k,r}U_k\right]\right]
  \]
  Therefore
  \begin{equation}\label{eq:diff-channel}
    (\Gamma^{\mathrm{acc}} - \widehat{\Gamma}^{\mathrm{acc}})(X) = \E_{k,r}\left[\Tr_\reg{T}\left[\left(\sum_{s \neq r} \Id_\reg{M} \otimes \ketbra{\psi^{(k,s)}}\right)U_k^\dagger\widetilde{X}_{k,r}U_k\right]\right]
  \end{equation}

  In this way, $\Gamma^{\mathrm{acc}} - \widehat{\Gamma}^{\mathrm{acc}}$ is the composition of three completely positive operations: (1) the channel which encrypts the input, applies $\mathcal{A}$, and reverses the encrypting unitary $U_k$; (2) the projection onto the subspace given by $\sum_{s \neq r} (\Id_\reg{M} \otimes \ketbra{\psi^{(k,s)}}_\reg{T})$; and (3) the partial trace.
  Because complete positivity is preserved through composition, $\Gamma^{\mathrm{acc}} - \widehat{\Gamma}^{\mathrm{acc}}$ is completely positive and $\widehat{\Gamma}^{\mathrm{acc}} \leq \Gamma^{\mathrm{acc}}$.
\end{proof}

\begin{proof}[Proof of \cref{lem:close-operators}]
  Let $\mathcal{Y} = \mathcal{X} - \mathcal{Z}$.
  By contractivity of the diamond norm under supermaps and the idempotence of $\Theta$, we can bound the contribution of $\Theta(\mathcal{Y})$ to the diamond norm as
  \[
    \| \Theta(\mathcal{Y}) \|_\diamond
    = \| \Theta(\mathcal{X}) - \Theta(\mathcal{Z})\|_\diamond
    = \| \Theta(\mathcal{X}) - \Theta^2(\mathcal{Z})\|_\diamond
    \leq \| \mathcal{X} - \Theta(\mathcal{Z})\|_\diamond
    \leq \varepsilon
  \]
  Our goal now is to find a `bridge' map $\mathcal{B}$ which is close to both $\mathcal{Y}$ and $\Theta(\mathcal{Y})$ in diamond distance.
  By \cref{thm:ksw}, there exist Stinespring contractions $V^\mathcal{X}$ and $V^{\Theta(\mathcal{Z})}$ along with an ancilla register $\reg{E}$ such that $(V^\mathcal{X},\reg{E})$ and $(V^{\Theta(\mathcal{Z})},\reg{E})$ are dilations of $\mathcal{X}$ and $\Theta(\mathcal{Z})$ respectively, and
  \[
    \| V^\mathcal{X} - V^{\Theta(\mathcal{Z})} \|_\infty \leq \sqrt{\|\mathcal{X} - \Theta(\mathcal{Z})\|_\diamond} \leq \sqrt{\varepsilon}
  \]
  Because $\mathcal{Y} \leq \mathcal{X}$, we can invoke \cref{thm:radon-nikodym} to assert the existence of an operator $0 \leq Q \leq \Id$ such that
  \[
    \mathcal{Y}(X) = \Tr_\reg{E}\left[(\Id_\reg{CD} \otimes Q^{\frac{1}{2}}) V^\mathcal{X} X{V^\mathcal{X}}^\dagger(\Id_\reg{CD} \otimes Q^{\frac{1}{2}})\right]
  \]
  Using these ingredients, we construct our bridge map as
  \[
    \mathcal{B}(X) = \Tr_\reg{E}\left[(\Id_\reg{CD} \otimes Q^{\frac{1}{2}})V^{\Theta(\mathcal{Z})} X{V^{\Theta(\mathcal{Z})}}^\dagger(\Id_\reg{CD} \otimes Q^{\frac{1}{2}})\right]
  \]
  Observe that $\mathcal{Y}$ and $\mathcal{B}$ have Stinespring operators $V^\mathcal{Y} = (\Id \otimes Q^{\frac{1}{2}})V^\mathcal{X}$ and $V^{\mathcal{B}} = (\Id \otimes Q^{\frac{1}{2}})V^{\Theta(\mathcal{Z})}$.
  Using \cref{thm:ksw}, we can use the operator distance between these to bound the diamond norm of their respective channels.
  \ifllncs
  \begin{align*}
    \| \mathcal{Y} - \mathcal{B} \|_\diamond
    &\leq 2\| V^\mathcal{Y} - V^{\mathcal{B}} \|_\infty \\
    &= 2 \| (\Id_\reg{CD} \otimes Q^{\frac{1}{2}})(V^\mathcal{X} - V^{\Theta(\mathcal{Z})})\|_\infty \\
    &\leq 2\|V^\mathcal{X} - V^{\Theta(\mathcal{Z})}\|_\infty \\
    &\leq 2\sqrt{\varepsilon}
  \end{align*}
  \else
  \[
    \| \mathcal{Y} - \mathcal{B} \|_\diamond \leq 2\| V^\mathcal{Y} - V^{\mathcal{B}} \|_\infty = 2 \| (\Id_\reg{CD} \otimes Q^{\frac{1}{2}})(V^\mathcal{X} - V^{\Theta(\mathcal{Z})})\|_\infty \leq 2\|V^\mathcal{X} - V^{\Theta(\mathcal{Z})}\|_\infty \leq 2\sqrt{\varepsilon}
  \]
  \fi
  where the second inequality relies on $Q$ being a contraction.
  In this way, $\mathcal{Y}$ and $\mathcal{B}$ are $2\sqrt{\varepsilon}$-close.
  By contractivity of the diamond norm under supermaps, we know that $\Theta(\mathcal{Y})$ and $\Theta(\mathcal{B})$ are also $2\sqrt{\varepsilon}$-close.

  Our final observation is that $\mathcal{B} \leq \Theta(\mathcal{Z})$ by \cref{thm:radon-nikodym}.
  Thus $\mathcal{B} \in \mathrm{M}$ (because $\mathrm{M}$ is downward-closed) and $\Theta(\mathcal{B}) = \mathcal{B}$.
  Putting all this together,
  \begin{align*}
    \| \mathcal{X} - \mathcal{Z} \|_\diamond
    &= \| \mathcal{Y} \|_\diamond \\
    &= \| \mathcal{Y} - \mathcal{B} + \mathcal{B}\|_\diamond \\
    &\leq \| \mathcal{Y} - \mathcal{B}\|_\diamond + \|\mathcal{B}\|_\diamond \\
    &= \| \mathcal{Y} - \mathcal{B}\|_\diamond + \|\Theta(\mathcal{B})\|_\diamond \\
    &= \| \mathcal{Y} - \mathcal{B}\|_\diamond + \|\Theta(\mathcal{B}) - \Theta(\mathcal{Y}) + \Theta(\mathcal{Y})\|_\diamond \\
    &\leq \| \mathcal{Y} - \mathcal{B}\|_\diamond + \|\Theta(\mathcal{B}) - \Theta(\mathcal{Y})\|_\diamond + \|\Theta(\mathcal{Y})\|_\diamond \\
    &\leq 4\sqrt{\varepsilon} + \varepsilon
  \end{align*}
\end{proof}

\fi

\subsection{Retrospective Security}
\label{sec:retrospective-security}

In order to make a data processing argument where the post-processing channel depends on the encryption randomness, we introduce a property for quantum encryption keys we call ``retrospective security.''

\begin{definition}[Retrospective Security]\label{def:retrospective-qca}
  Let $\mathsf{QES} = (\Gen, \Enc, \Dec)$ be a SKQES and define the parameterized channel $\Channel{\Reveal_{k,r}}{M \oplus \bot}{MK}$ as
  \[
    \Reveal_{k,r}(\rho_\reg{M}) = (\Id - \Pi^\bot) \rho (\Id - \Pi^\bot) \otimes \ketbra{k,r}_\reg{K} + \Pi^\bot \rho \Pi^\bot \otimes \ketbra{\bot}_\reg{K}.
  \]
  Then $\mathsf{QES}$ is QCA with {\em retrospective security} if, for all channels $\Channel{A}{CR}$, there exists a completely positive map $\Channel{\Lambda^{\mathrm{rej}}_{\mathrm{R}}}{R}$ such that $\|\Gamma_{\mathrm{R}} - \Lambda_{\mathrm{R}}\|_\diamond \leq \negl(n)$ and $\Lambda_{\mathrm{R}}$ is trace-preserving where
  \ifllncs
  \begin{align*}
    \Gamma_{\mathrm{R}} &= \E_{k,r}\left[\Reveal_{k,r} \circ \Dec_k \circ \mathcal{A} \circ \Enc_{k;r}\right] \\
    \Lambda_{\mathrm{R}} &= \Id_\reg{M} \otimes \E_{k,r}[\ketbra{k,r}_\reg{K} \otimes \Lambda^{\mathrm{acc}}_{k,r}] + \ketbra{\bot}_\reg{MK} \otimes \Lambda^{\mathrm{rej}}_{\mathrm{R}}
  \end{align*}
  \else
  \[
    \Gamma_{\mathrm{R}} = \E_{k,r}\left[\Reveal_{k,r} \circ \Dec_k \circ \mathcal{A} \circ \Enc_{k;r}\right]
    \eqand{1.5em}
    \Lambda_{\mathrm{R}} = \Id_\reg{M} \otimes \E_{k,r}[\ketbra{k,r}_\reg{K} \otimes \Lambda^{\mathrm{acc}}_{k,r}] + \ketbra{\bot}_\reg{MK} \otimes \Lambda^{\mathrm{rej}}_{\mathrm{R}}
  \]
  \fi
  given that $\Lambda^{\mathrm{acc}}_{k,r}$ is as in \cref{def:qca}.
\end{definition}

Before we show that all QCA-secure schemes satisfy this definition, we first show that the randomness preservation of \cref{lem:preserve-randomness} extends to the retrospective setting where $(k,r)$ are included in the output.

\begin{corollary} \label{cor:preserve-randomness-retro}
  Let $(\Gen, \Enc, \Dec)$ be a SKQES, $\mathcal{A} : \reg{CR} \to \reg{CR}$ be a channel, $\Gamma_\mathrm{R}$ be as in \cref{def:retrospective-qca}, and $\Gamma_{\mathrm{R}}^\mathrm{acc}$ be the acceptance branch of that map.
  Let $\widehat{\Gamma}_{\mathrm{R}} = \E_{k,r}[\Reveal_{k,r} \circ \Dec_{k;r} \circ \mathcal{A} \circ \Enc_{k;r}]$ where $\Dec_{k;r}$ is as in \cref{lem:preserve-randomness}.
  If $(\Gen, \Enc, \Dec)$ is QCA secure then $\| \Gamma_{\mathrm{R}}^{\mathrm{acc}} - \widehat{\Gamma}_{\mathrm{R}}^{\mathrm{acc}} \|_\diamond \leq \negl(n)$.
\end{corollary}
\begin{proof}
  By following the same argument as in \cref{lem:preserve-randomness}, it can be seen that $\widehat{\Gamma}_{\mathrm{R}}^{\mathrm{acc}} \leq \Gamma_{\mathrm{R}}^{\mathrm{acc}}$ and therefore $\Gamma_{\mathrm{R}}^{\mathrm{acc}} - \widehat{\Gamma}_{\mathrm{R}}^{\mathrm{acc}}$ is completely positive.
  By \cref{fact:diamond-norm-partial-trace} and \cref{lem:preserve-randomness},
  \[
    \| \Gamma_{\mathrm{R}}^{\mathrm{acc}} - \widehat{\Gamma}_{\mathrm{R}}^{\mathrm{acc}} \|_\diamond = \| \Tr_\reg{K} \circ \Gamma_{\mathrm{R}}^{\mathrm{acc}} - \Tr_\reg{K} \circ \widehat{\Gamma}_{\mathrm{R}}^{\mathrm{acc}} \|_\diamond = \| \Gamma^{\mathrm{acc}} - \widehat{\Gamma}^{\mathrm{acc}} \|_\diamond \leq \negl(n). \qedhere
  \]
\end{proof}

\begin{lemma} \label{lem:qca-implies-retrospective}
  If $(\Gen, \Enc, \Dec)$ is an SKQES with QCA security then it has retrospective security as well.
\end{lemma}

We begin by stating an operator-theoretic technical lemma about ``classical extensions'' of quantum channels.
Given a channel $\Channel{N}{A}{B}$, we could describe a classical extension of $\mathcal{N}$ as any channel $\Channel{N'}{A}{BC}$ such that $\Tr_\reg{C} \circ \mathcal{N}' = \mathcal{N}$ and the output of $\mathcal{N}'$ on $\reg{C}$ is always classical.
In \cref{lem:classical-extensions}, we show that given any channels $\mathcal{X}$ and $\mathcal{Y}$ which are close (in diamond distance) and any classical extension of $\mathcal{X}$, there exists some corresponding classical extension of $\mathcal{Y}$ such that the extensions are also close.

\ifllncs
We are also going to work with supermaps which are {\em retractions} in that their output always lies within a particular subspace on which the supermap acts trivially.
A retraction is analogous to a projection: it is idempotent, has image $\mathrm{M}$, and fixes every element of $\mathrm{M}$.

\begin{definition}
  Let $\mathrm{M} \subseteq \CP{A}$ be a subset of completely positive maps and $\Map{\Theta}{A}{A}$ be a supermap. Then $\Theta$ is a {\em retraction onto $\mathrm{M}$} if $\Theta(\mathcal{N}) \in \mathrm{M}$ for all $\mathcal{N} \in \CP{A}$ and $\Theta(\mathcal{M}) = \mathcal{M}$ for all $\mathcal{M} \in \mathrm{M}$.
\end{definition}
\fi

\begin{lemma}\label{lem:classical-extensions}
  Let $\Map{\Theta}{A,B}{A,B}$ be a retraction onto a downward-closed subset $\mathrm{M} \subseteq \CP{A}{B}$ of completely positive maps and let $\Channel{X,Y}{A}{B}$ be completely positive, trace-non{\ifllncs-\fi}increasing maps such that $\Theta(\mathcal{Y}) = \mathcal{Y}$ and $\|\mathcal{X} - \mathcal{Y}\|_\diamond \leq \varepsilon$.
  Additionally, let $\Channel{X'}{A}{BC}$ be a classical extension of $\mathcal{X}$ in that $\mathcal{X}' = \sum_c \mathcal{X}_c \otimes \ketbra{c}_\reg{C}$ where $\sum_c \mathcal{X}_c = \mathcal{X}$.

  Then there exists $\Channel{Y'}{A}{BC}$ such that $(\Theta \otimes \Id_\reg{C})(\mathcal{Y}') = \mathcal{Y}'$ and $\|\mathcal{X}' - \mathcal{Y}'\|_\diamond \leq 2\sqrt{\varepsilon}$.
\end{lemma}

Before we prove \cref{lem:classical-extensions}, we give motivation by showing that it implies \cref{lem:qca-implies-retrospective}.
Given our context, the retrospective effective attack maps can be understood as classical extensions of the original QCA attack maps.
In this way, we can use \cref{lem:classical-extensions} to show that there exists some corresponding extension of $\Lambda$ which must be close to $\Gamma_{\mathrm{R}}$.
The rest of the proof is showing that extension is close to the constructed $\Lambda_{\mathrm{R}}$.

\ifllncs
The supermap $\Theta$ will be the ``Bell-projection'' described in \cref{fact:bell-projection} which retracts onto the subspace of maps which act trivially on the first register.
This supermap is implicit in the QCA definition as part of the transformation from the real channel to the simulated channel.
We state and prove two facts to align this supermap and subspace with the requirements of \cref{lem:classical-extensions}, give a proof of \cref{lem:qca-implies-retrospective} assuming \cref{lem:classical-extensions}, and then prove \cref{lem:classical-extensions}.

\begin{fact} \label{fact:bell-projection}
  The ``Bell-projection'' supermap $\Map{\Theta}{AB}{AC}$ given by
  \[
    \Theta(\mathcal{N})(\rho) = \Tr_\reg{ZA'}[(\Phi^+_\reg{ZA'} \otimes \Id_\reg{AC})(\mathcal{N}_{\reg{A'B}\to\reg{A'C}} \otimes \Id_\reg{ZA})(\Phi^+_\reg{ZA'} \otimes \rho_\reg{AB})].
  \]
  is a retraction onto the subspace $\mathrm{M} = \{ \Id_\reg{A} \otimes \mathcal{B} :  \mathcal{B} \in \LinearOps{B}{C} \} \subset \LinearOps{AB}{AC}$.
\end{fact}
\begin{proof}
  Clearly $\Theta(\mathcal{N})$ acts as the identity on $\reg{A}$ by construction for any $\mathcal{N} \in \LinearOps{AB}{AC}$.
  Additionally $\Theta$ acts trivially upon any $\mathcal{N} = \Id_\reg{A} \otimes \mathcal{B}$ for $\mathcal{B} \in \LinearOps{B}{C}$.
  \begin{align*}
    \Theta(\Id_\reg{A} \otimes \mathcal{B})(\rho)
    &= \Tr_\reg{ZA'}[(\Phi^+_\reg{ZA'} \otimes \Id_\reg{AC})(\Id_\reg{ZAA'} \otimes \mathcal{B}_{\reg{B}\to\reg{C}})(\Phi^+_\reg{ZA'} \otimes \rho_\reg{AB})] \\
    &= \Tr_\reg{ZA'}[(\Id_\reg{ZAA'} \otimes \mathcal{B}_{\reg{B}\to\reg{C}})(\Phi^+_\reg{ZA'} \otimes \rho_\reg{AB})] \\
    &= (\Id_\reg{A} \otimes \mathcal{B}_{\reg{B}\to\reg{C}})(\rho_\reg{AB}) \qedhere
  \end{align*}
\end{proof}

\begin{fact} \label{fact:downward-closed}
  {\unless\ifllncs The subspace \fi}$\mathrm{M} = \{ \Id_\reg{A} \otimes \mathcal{B} : \mathcal{B} \in \CP{B}{C} \} \subset \CP{AB}{AC}$ is downward-closed.
\end{fact}
\begin{proof}
  Let $\mathcal{X} = \Id_\reg{A} \otimes \mathcal{B} \in \mathrm{M}$ be a map in the subspace where $\mathcal{B} \in \LinearOps{B}{C}$ and $\mathcal{Z}$ be some map such that $\mathcal{Z} \leq \mathcal{X}$.
  Additionally, let $(V^\mathcal{B}, \mathrm{E})$ be a Stinespring dilation of $\mathcal{B}$ and note that $(\Id \otimes V^{\mathcal{B}}, \reg{E})$ is a Stinespring dilation of $\mathcal{X}$.
  By \cref{thm:radon-nikodym}, there exists $0 \leq Q \leq \Id_\reg{E}$ such that
  \[
    \mathcal{Z}(X) = \Tr_\reg{E}[(\Id_\reg{AC} \otimes Q^{\frac{1}{2}})(\Id_\reg{A} \otimes V^{\mathcal{B}}) X (\Id_\reg{A} \otimes V^{\mathcal{B}})^\dagger(\Id_\reg{AC} \otimes Q^{\frac{1}{2}})]
  \]
  which acts trivially on the $\reg{A}$ register. Therefore $\mathcal{Z} \in \mathrm{M}$ and $\mathrm{M}$ is downward-closed.
\end{proof}
\fi

\begin{proof}[Proof of \cref{lem:qca-implies-retrospective}]
  Fix an attack $\mathcal{A}$ and let $\Gamma$ and $\Lambda$ be the corresponding real and simulated effective attack maps given by the scheme's QCA security.
  Observe that $\Dec$ performs an explicit accept-or-reject measurement so its output is a classical mixture of the two orthogonal cases.
  Thus the real map can be decomposed as $\Gamma = \Gamma^{\mathrm{acc}} + \Gamma^{\mathrm{rej}}$ where $\Gamma^{\mathrm{acc}}(\rho) = (\Id - \Pi^\bot)\Gamma(\rho)(\Id - \Pi^\bot)$ and $\Gamma^{\mathrm{rej}}(\rho) = \Pi^\bot\Gamma(\rho)\Pi^\bot$.
  By data processing,
  \begin{equation} \label{eq:acc-rej-both-negl}
    \Big\|\Gamma^{\mathrm{acc}} - \Id_\reg{M} \otimes \E_{k,r}[\Lambda^{\mathrm{acc}}_{k,r}]\Big\|_\diamond \leq \negl(n)
    \ifllncs \eqand{4mm} \else \eqand{1cm} \fi
    \Big\|\Gamma^{\mathrm{rej}} - \ketbra{\bot}_\reg{M} \otimes \Lambda^{\mathrm{rej}}\Big\|_\diamond \leq \negl(n).
  \end{equation}
  In the same way, we can define the real retrospective attack map $\Gamma_{\mathrm{R}} = \Gamma_{\mathrm{R}}^{\mathrm{acc}} + \Gamma_{\mathrm{R}}^{\mathrm{rej}}$ where $\Gamma^{\mathrm{acc}}_{\mathrm{R}}(\rho) = {(\Id - \Pi^\bot)}\Gamma_{\mathrm{R}}(\rho)(\Id - \Pi^\bot)$ and $\Gamma^{\mathrm{rej}}_{\mathrm{R}}(\rho) = \Pi^\bot\Gamma_{\mathrm{R}}(\rho)\Pi^\bot$, and the simulated retrospective attack map $\Lambda_{\mathrm{R}}$ where $\Lambda^{\mathrm{rej}}_{\mathrm{R}} = \Lambda^{\mathrm{rej}}$.
  By the triangle inequality,
  \[
    \| \Gamma_\mathrm{R} - \Lambda_{\mathrm{R}} \|_\diamond
    \leq \| \Gamma_{\mathrm{R}}^{\mathrm{acc}} - \Id_\reg{M} \otimes \E_{k,r}[\Lambda_{k,r}^{\mathrm{acc}} \otimes \ketbra{k,r}_\reg{K}] \|_\diamond + \| \Gamma_{\mathrm{R}}^{\mathrm{rej}} - \ketbra{\bot}_\reg{MK} \otimes \Lambda_{\mathrm{R}}^{\mathrm{rej}} \|_\diamond
  \]
  Because $\Omega_{\mathrm{R}}^{\mathrm{rej}} = (\Id_\reg{M} \otimes \ketbra{\bot}_\reg{K}) \circ \Omega^{\mathrm{rej}}$ for $\Omega \in \{\Gamma, \Lambda\}$ --- i.e., the rejection case just composes with a fixed state preparation channel --- we can easily bound
  \ifllncs
    \begin{align*}
      \| \Gamma_{\mathrm{R}}^{\mathrm{rej}} - \ketbra{\bot}_{\reg{MK}}  \otimes \Lambda_{\mathrm{R}}^{\mathrm{rej}} \|_\diamond
      &\leq \| (\Id_\reg{MR} \otimes \ketbra{\bot}_\reg{K}) \circ (\Gamma^{\mathrm{rej}} - \ketbra{\bot}_{\mathsf{M}} \otimes \Lambda^{\mathrm{rej}}) \|_\diamond \\
      &\leq \| \Gamma^{\mathrm{rej}} - \ketbra{\bot}_{\reg{M}} \otimes \Lambda^{\mathrm{rej}} \|_\diamond \\
      &\leq \negl(n).
    \end{align*}
  \else
  \[
    \| \Gamma_{\mathrm{R}}^{\mathrm{rej}} - \ketbra{\bot}_{\reg{MK}}  \otimes \Lambda_{\mathrm{R}}^{\mathrm{rej}} \|_\diamond
    \leq \| (\Id_\reg{MR} \otimes \ketbra{\bot}_\reg{K}) \circ (\Gamma^{\mathrm{rej}} - \ketbra{\bot}_{\mathsf{M}} \otimes \Lambda^{\mathrm{rej}}) \|_\diamond
    \leq \| \Gamma^{\mathrm{rej}} - \ketbra{\bot}_{\reg{M}} \otimes \Lambda^{\mathrm{rej}} \|_\diamond
    \leq \negl(n).
  \]
  \fi

  The difficulty is in bounding the acceptance term.
  The main tool for that will be invoking \cref{lem:classical-extensions} where $\mathcal{X} = \Gamma^{\mathrm{acc}}$, $\mathcal{X}' = \Gamma_\mathrm{R}^{\mathrm{acc}}$, $\mathcal{Y} = \Id_\reg{M} \otimes \E_{k,r}[\Lambda_{k,r}^{\mathrm{acc}}]$, and $\Theta$ is the Bell-projection supermap.
  \cref{eq:acc-rej-both-negl} has established that $\| \mathcal{X} - \mathcal{Y} \|_\diamond \leq \negl(n)$, Facts \ref{fact:bell-projection} and \ref{fact:downward-closed} prove that $\Theta$ is a retraction onto the downward-closed subset $\mathrm{M} = \{ \Id_\reg{M} \otimes \mathcal{B} : \mathcal{B} \in \LinearOps{R} \}$, and clearly $\mathcal{Y} \in \mathrm{M}$.
  Thus the lemma yields a $\mathcal{Y}'$ such that $\| \Gamma_{\mathrm{R}}^{\mathrm{acc}} - \mathcal{Y}'\|_\diamond \leq \negl(n)$ and $\Theta(\mathcal{Y}') = \mathcal{Y}'$ which can be used as follows.
  \begin{align*}
    \| \Gamma_{\mathrm{R}}^{\mathrm{acc}} - \Lambda_{\mathrm{R}}^{\mathrm{acc}} \|_\diamond
    &\leq \| \Gamma_{\mathrm{R}}^{\mathrm{acc}} - \mathcal{Y}' + \mathcal{Y}' - \Theta(\Gamma_{\mathrm{R}}^{\mathrm{acc}}) + \Theta(\Gamma_{\mathrm{R}}^{\mathrm{acc}}) - \Lambda_{\mathrm{R}}^{\mathrm{acc}} \|_\diamond \\
    &\leq \| \Gamma_{\mathrm{R}}^{\mathrm{acc}} - \mathcal{Y}' \|_\diamond + \| \mathcal{Y}' - \Theta(\Gamma_{\mathrm{R}}^{\mathrm{acc}}) \|_\diamond + \| \Theta(\Gamma_{\mathrm{R}}^{\mathrm{acc}}) - \Lambda_{\mathrm{R}}^{\mathrm{acc}} \|_\diamond \\
    &= \| \Gamma_{\mathrm{R}}^{\mathrm{acc}} - \mathcal{Y}' \|_\diamond + \| \Theta(\mathcal{Y}') - \Theta(\Gamma_{\mathrm{R}}^{\mathrm{acc}}) \|_\diamond + \| \Theta(\Gamma_{\mathrm{R}}^{\mathrm{acc}}) - \Theta(\widehat{\Gamma}_{\mathrm{R}}^{\mathrm{acc}}) \|_\diamond \\
    &\leq 2\| \Gamma_{\mathrm{R}}^{\mathrm{acc}} - \mathcal{Y}' \|_\diamond + \| \Gamma_{\mathrm{R}}^{\mathrm{acc}} - \widehat{\Gamma}_{\mathrm{R}}^{\mathrm{acc}} \|_\diamond \\
    &\leq \negl(n)
  \end{align*}
  where the equality used the identity $\Lambda_{\mathrm{R}}^{\mathrm{acc}} = \Theta(\widehat{\Gamma}_{\mathrm{R}}^{\mathrm{acc}})$ and the final inequality follows from \cref{cor:preserve-randomness-retro}.
\end{proof}

\begin{proof}[Proof of \cref{lem:classical-extensions}]
  Let $(V^\mathcal{X},\reg{E})$ be a Stinespring dilation of $\mathcal{X}'$ so that $\mathcal{X}'(\rho) = \Tr_\reg{E}[V^\mathcal{X}\rho(V^\mathcal{X})^\dagger]$.
  Note that $\mathcal{X}(\rho) = \Tr_\reg{CE}[V^\mathcal{X}\rho(V^\mathcal{X})^\dagger]$ because $\Tr_\reg{C} \circ \mathcal{X}' = \mathcal{X}$, and therefore $(V^\mathcal{X},\reg{CE})$ is a valid Stinespring dilation of $\mathcal{X}$ as well.
  By \cref{thm:ksw}, there exists an operator $V^\mathcal{Y}$ such that $(V^\mathcal{Y},\reg{CE})$ is a Stinespring dilation of $\mathcal{Y}$ (expanding $\reg{E}$ if necessary) and
  \[
    \| V^\mathcal{X} - V^\mathcal{Y} \|_\infty \leq \sqrt{\|\mathcal{X} - \mathcal{Y}\|_\diamond} \leq \sqrt{\varepsilon}
  \]
  Let $\mathcal{Y}'_0(\rho) = \Tr_\reg{E}[V^\mathcal{Y}\rho(V^\mathcal{Y})^\dagger]$ be the natural extension of $\mathcal{Y}$ to $\reg{C}$ suggested by its dilation $V^\mathcal{Y}$ and $\mathcal{Y}' = \Delta_\reg{C} \circ \mathcal{Y}'_0$ where we denote by $\Delta_\reg{C}(\rho) = \sum_c \ketbra{c}\rho\ketbra{c}$ the dephasing channel on the $\reg{C}$ register.
  Because $\mathcal{X}'$ is a classical extension, $\Delta_\reg{C} \circ \mathcal{X}' = \mathcal{X}'$ and thus
  \[
    \| \mathcal{X}' - \mathcal{Y}' \|_\diamond
    = \| \Delta_\reg{C} \circ (\mathcal{X}' - \mathcal{Y}'_0) \|_\diamond
    \leq \| \mathcal{X}' - \mathcal{Y}'_0 \|_\diamond
    \leq 2\| V^\mathcal{X} - V^\mathcal{Y} \|_\infty
    \leq 2\sqrt{\varepsilon}
  \]
  where we've invoked the upper bound from \cref{thm:ksw}.

  It remains to be shown that $(\Theta \otimes \Id_\reg{C})(\mathcal{Y}') = \mathcal{Y}'$.
  If we define $\mathcal{Y}_c(\rho) = (\Id_\reg{B} \otimes \bra{c}_\reg{C})\mathcal{Y}'(\rho)(\Id_\reg{B} \otimes \ket{c}_\reg{C})$, then we can decompose $\mathcal{Y}'(\rho) = \sum_c \mathcal{Y}_c(\rho) \otimes \ketbra{c}_\reg{C}$.
  Because $\Tr_\reg{C} \circ \mathcal{Y}' = \mathcal{Y}$, we have that $\mathcal{Y} = \sum_c \mathcal{Y}_c$ and therefore $\mathcal{Y}_c \leq \mathcal{Y}$ for all $c$.
  $\mathrm{M}$ is downward-closed and $\mathcal{Y} \in \mathrm{M}$ so $\mathcal{Y}_c \in \mathrm{M}$ for all $c$.
  Thus we arrive at
  \[
    (\Theta \otimes \Id_\reg{C})(\mathcal{Y}')(\rho) = \sum_c \Theta(\mathcal{Y}_c)(\rho) \otimes \ketbra{c} = \sum_c \mathcal{Y}_c(\rho) \otimes \ketbra{c} = \mathcal{Y}'(\rho)
  \]
\end{proof}

\subsection{Deferral Lemma}

\begin{proof}[Proof of \cref{lem:exps-close}]

Fix $\mathcal{A}$.
Then we prove our result through a hybrid argument starting with $\Hyb_1$.

\begin{center}
\fbox{%
  \begin{minipage}{0.9\linewidth}
    $\Hyb_1(\rho)$
    \begin{enumerate}
      \item Sample $k \gets \Gen(1^n)$ and $r \gets p_k$.
      \item Prepare $\mathcal{A} \circ \Enc_{k;r}(\rho_\reg{\changed{M}R})$
      \item Measure $\changed{\reg{C}}$ with \changed{$\{\Pi_k, \Id - \Pi_k\}$} and post-select onto \changed{$\Pi_k$}.
      \item \changed{Decrypt $\reg{C}$ with $\Dec_k$ yielding register $\reg{M'}$.}
      \item \changed{Encrypt $\reg{M'}$ with $\Enc_{k;r}$ yielding register $\reg{C'}$.}
      \item Output $\reg{C'R}$ and $k$.
    \end{enumerate}
  \end{minipage}
}
\end{center}

\begin{claim} \label{claim:retro-hyb-0-vs-1}
  $\| \Exp_\mathsf{defer}^\mathcal{A} - \Hyb_1\|_\diamond \leq \negl(n)$
\end{claim}
\begin{proof}
  Observe that $\Hyb_1$ and $\Exp_\mathsf{defer}^\mathcal{A}$ resemble the real and simulated effective attack maps from \cref{def:retrospective-qca}.
  Our goal then is to show that the retrospective security of the SKQES implies our bound.
  Let $\Gamma_{\mathrm{R}}^{\mathrm{acc}}$ and $\Lambda_{\mathrm{R}}^{\mathrm{acc}}$ be the acceptance branches of the real and simulated channels from \cref{def:retrospective-qca}, respectively.
  By \cref{lem:qca-implies-retrospective}, $\| \Gamma_{\mathrm{R}}^{\mathrm{acc}} - \Lambda_{\mathrm{R}}^{\mathrm{acc}} \|_\diamond \leq \negl(n)$.
  Then we define a post-processing channel $\mathcal{N}$ which acts on the output of $\Gamma_{\mathrm{R}}^{\mathrm{acc}}$ or $\Lambda_{\mathrm{R}}^{\mathrm{acc}}$ in the following way:
  \begin{enumerate}
    \item Measure $\reg{M}$ using $\{\Pi^{\bot}, \Id - \Pi^{\bot}\}$. Post-select onto $\Id - \Pi^{\bot}$.
    \item Measure $\reg{K}$ in the standard basis to get $(k, r)$.
    \item Re-encrypt the $\reg{M}$ register with $\Enc_{k;r}$ yielding $\reg{C'}$.
    \item Output $\reg{C'R}$ and $k$.
  \end{enumerate}
  By \cref{fact:dec-before-after-proj} (and an analogous fact for the $\widetilde{\Pi}_{k,r}$ projector and the simulated decryption channel), we have that $\Exp_\mathsf{defer}^\mathcal{A} = \mathcal{N} \circ \Lambda_{\mathrm{R}}^{\mathrm{acc}}$ and $\Hyb_1 = \mathcal{N} \circ \Gamma_{\mathrm{R}}^{\mathrm{acc}}$.
  Thus we have our desired bound,
  \[
    \| \Exp_\mathsf{defer}^\mathcal{A} - \Hyb_1\|_\diamond
    = \| \mathcal{N} \circ \Gamma_{\mathrm{R}}^{\mathrm{acc}} - \mathcal{N} \circ \Lambda_{\mathrm{R}}^{\mathrm{acc}} \|_\diamond
    \leq \| \Gamma_{\mathrm{R}}^{\mathrm{acc}} - \Lambda_{\mathrm{R}}^{\mathrm{acc}} \|_\diamond \leq \negl(n) \qedhere
  \]
\end{proof}

\begin{center}
\fbox{%
  \begin{minipage}{0.9\linewidth}
    $\Hyb_2(\rho)$
    \begin{enumerate}
      \item Sample $k \gets \Gen(1^n)$ and $r \gets p_k$.
      \item Prepare $\mathcal{A} \circ \Enc_{k;r}(\rho_\reg{MR})$
      \item Measure $\reg{C}$ with $\{\Pi_{\changed{k,r}}, \Id - \Pi_{\changed{k,r}}\}$ and post-select onto $\Pi_{\changed{k,r}}$. \label{item:retro-hyb-2-post-select}
      \item Decrypt $\reg{C}$ with $\Dec_k$ yielding register $\reg{M'}$. \label{item:retro-hyb-2-decrypt}
      \item Encrypt $\reg{M'}$ with $\Enc_{k;r}$ yielding register $\reg{C'}$. \label{item:retro-hyb-2-encrypt}
      \item Output $\reg{C'R}$ and $k$.
    \end{enumerate}
  \end{minipage}
}
\end{center}

\begin{claim} \label{claim:retro-hyb-1-vs-2}
  $\| \Hyb_1 - \Hyb_2 \|_\diamond \leq \negl(n)$
\end{claim}
\begin{proof}
  Let $\Gamma_{\mathrm{R}}^{\mathrm{acc}}$ and $\widehat{\Gamma}_{\mathrm{R}}^{\mathrm{acc}}$ be the subchannels as in \cref{cor:preserve-randomness-retro} for the fixed attack $\mathcal{A}$, and let $\mathcal{N}$ be the post-processing map which does the following:
  \begin{enumerate}
    \item Measure $\reg{K}$ to get outcome $(k,r)$.
    \item Apply $\Enc_{k;r}$ to $\reg{M'}$ to get $\reg{C'}$.
    \item Output $\reg{C'R}$ and $k$.
  \end{enumerate}
  Similar to \cref{fact:dec-before-after-proj}, notice that $\Dec_k(\Pi_{k,r}\rho\Pi_{k,r}) = (\Id - \Pi^\bot)\Dec_{k;r}(\rho)(\Id - \Pi^\bot)$.
  Therefore $\Hyb_1 = \mathcal{N} \circ \Gamma_{\mathrm{R}}^{\mathrm{acc}}$ and $\Hyb_2 = \mathcal{N} \circ \widehat{\Gamma}_{\mathrm{R}}^{\mathrm{acc}}$.
  By \cref{cor:preserve-randomness-retro},
  \[
    \| \Hyb_1 - \Hyb_2 \|_\diamond
    = \|\,\mathcal{N} \circ \widehat{\Gamma}_{\mathrm{R}}^{\mathrm{acc}} - \mathcal{N} \circ \Gamma_{\mathrm{R}}^{\mathrm{acc}} \|_\diamond
    \leq \| \widehat{\Gamma}_{\mathrm{R}}^{\mathrm{acc}} - \Gamma_{\mathrm{R}}^{\mathrm{acc}} \|_\diamond
    \leq \negl(n) \qedhere
  \]
\end{proof}

\begin{claim} \label{claim:retro-hyb-2-vs-3}
  $\| \Hyb_2 - \Exp_\mathsf{real}^\mathcal{A} \|_\diamond \leq \negl(n)$
\end{claim}
\begin{proof}
  First notice that Steps \ref{item:retro-hyb-2-decrypt} and \ref{item:retro-hyb-2-encrypt} act as the identity in $\Hyb_2$ (because of the post-selection in Step \ref{item:retro-hyb-2-post-select}) and so can be omitted.
  Then the only remaining difference is in projectors --- i.e., $\Pi_{k,r}$ versus $\Pi_k$.
  Therefore we want to prove for all states $\rho_\reg{MRE}$, where $\reg{E}$ is an arbitrary reference register on which all maps act trivially,
  \ifllncs
  \begin{align*}
    &\left\|\E_{k,r}\left[\Pi_k(\mathcal{A} \circ \Enc_{k;r}(\rho))\Pi_k \otimes \ketbra{k}_\reg{K}\right]\right. \\
    &\qquad - \left.\E_{k,r}\left[\Pi_{k,r}(\mathcal{A} \circ \Enc_{k;r}(\rho))\Pi_{k,r} \otimes \ketbra{k}_\reg{K}\right]\right\|_1 \leq \negl(n).
  \end{align*}
  \else
  \[
    \left\|\E_{k,r}\left[\Pi_k(\mathcal{A} \circ \Enc_{k;r}(\rho))\Pi_k \otimes \ketbra{k}_\reg{K}\right]
    - \E_{k,r}\left[\Pi_{k,r}(\mathcal{A} \circ \Enc_{k;r}(\rho))\Pi_{k,r} \otimes \ketbra{k}_\reg{K}\right]\right\|_1 \leq \negl(n).
  \]
  \fi
  Fix $\rho$ and let $\sigma_{k,r} = \mathcal{A} \circ \Enc_{k;r}(\rho)$ and $\Pi_{k,\neq r} = \sum_{s \neq r} \Pi_{k,s}$.
  Then we seek to bound the trace norm of
  \begin{align*}
    &\E_{k,r}[(\Pi_k\sigma_{k,r}\Pi_k - \Pi_{k,r}\sigma_{k,r}\Pi_{k,r}) \otimes \ketbra{k}_\reg{K}] \\
    &\qquad = \E_{k,r}\left[\left(\Pi_{k,\neq r}\sigma_{k,r}\Pi_{k,\neq r}
    + \Pi_{k,r}\sigma_{k,r}\Pi_{k,\neq r}
    + \Pi_{k,\neq r}\sigma_{k,r}\Pi_{k,r}\right) \otimes \ketbra{k}_\reg{K}\right]
  \end{align*}
  Adapting \cref{eq:diff-channel} to the retrospective setting and observing $\sum_{s \neq r} \Id_\reg{M} \otimes \ketbra{\psi^{(k,s)}} = U_k^\dagger\Pi_{k,\neq r}U_k$,
  \ifllncs
  \begin{align*}
    &\| (\Gamma_{\mathrm{R}}^{\mathrm{acc}} - \widehat{\Gamma}_{\mathrm{R}}^{\mathrm{acc}})(\rho)\|_1 \\
    &\qquad= \left\| \E_{k,r} \Tr_\reg{T}\left[\left(\sum_{s \neq r} \Id_\reg{M} \otimes \ketbra{\psi^{(k,s)}}\right) U_k^\dagger\sigma_{k,r}U_k\right] \otimes \ketbra{k,r}_\reg{K} \right\|_1 \\
    &\qquad= \left\| \E_{k,r} \Tr_\reg{T}\left[U_k^\dagger \Pi_{k,\neq r}\sigma_{k,r}\Pi_{k,\neq r}U_k\right] \otimes \ketbra{k,r}_\reg{K} \right\|_1 \\
    &\qquad= \left\| \E_{k,r} \left[\Pi_{k,\neq r} \sigma_{k,r} \Pi_{k,\neq r} \otimes \ketbra{k}_\reg{K}\right]\right\|_1
  \end{align*}
  \else
  \begin{align*}
    \| (\Gamma_{\mathrm{R}}^{\mathrm{acc}} - \widehat{\Gamma}_{\mathrm{R}}^{\mathrm{acc}})(\rho)\|_1
    &= \left\| \E_{k,r} \Tr_\reg{T}\left[\left(\sum_{s \neq r} \Id_\reg{M} \otimes \ketbra{\psi^{(k,s)}}\right) U_k^\dagger\sigma_{k,r}U_k\right] \otimes \ketbra{k,r}_\reg{K} \right\|_1 \\
    &= \left\| \E_{k,r} \Tr_\reg{T}\left[U_k^\dagger \Pi_{k,\neq r}\sigma_{k,r}\Pi_{k,\neq r}U_k\right] \otimes \ketbra{k,r}_\reg{K} \right\|_1 \\
    &= \left\| \E_{k,r} \left[\Pi_{k,\neq r} \sigma_{k,r} \Pi_{k,\neq r} \otimes \ketbra{k}_\reg{K}\right]\right\|_1
  \end{align*}
  \fi
  The last equality holds because both operators are positive and have the same trace.
  By \cref{cor:preserve-randomness-retro}, we know that $\|\E_{k,r}\left[\Pi_{k,\neq r}\sigma_{k,r}\Pi_{k,\neq r} \otimes \ketbra{k}_\reg{K}\right]\|_1 \leq \negl(n)$, but the cross terms are not accounted for in that setting because decryption traces out the tag register $\reg{T}$.
  By H\"{o}lder's inequality,
  \[
    \| \Pi_{k,r}\sigma_{k,r}\Pi_{k,\neq r} \|_1
    \leq \| \Pi_{k,r}\sqrt{\sigma_{k,r}} \|_2 \cdot \|\sqrt{\sigma_{k,r}}\,\Pi_{k,\neq r}\|_2
    \leq \sqrt{\Tr[\Pi_{k,\neq r}\sigma_{k,r} \Pi_{k,\neq r}]}
  \]
  By convexity of the trace norm, $\|X \otimes \ketbra{k}\|_1 = \|X\|_1$, and Jensen's inequality, we have that
  \begin{align*}
    \left\| \E_{k,r}[\Pi_{k,r}\sigma_{k,r}\Pi_{k,\neq r} \otimes \ketbra{k}_\reg{K}] \right\|_1
    &\leq \E_{k,r} \left\|\Pi_{k,r}\sigma_{k,r}\Pi_{k,\neq r}\right \|_1 \\
    &\leq \sqrt{\E_{k,r} \Tr[\Pi_{k,\neq r}\sigma_{k,r} \Pi_{k,\neq r}]} \\
    &= \sqrt{\E_{k,r} \Tr[\Pi_{k,\neq r}\sigma_{k,r} \Pi_{k,\neq r} \otimes \ketbra{k}_\reg{K}]} \\
    &= \left\|\E_{k,r}[\Pi_{k,\neq r}\sigma_{k,r} \Pi_{k,\neq r} \otimes \ketbra{k}_\reg{K}]\right\|_1^{\frac{1}{2}}\\
    &\leq \negl(n)
  \end{align*}
  The same argument can be applied to the $\Pi_{k,\neq r}\sigma_{k,r}\Pi_{k,r}$ term.
  Because all three terms are negligible, our claim follows directly from the triangle inequality.
\end{proof}

\noindent Combining Claims \ref{claim:retro-hyb-0-vs-1} -- \ref{claim:retro-hyb-2-vs-3} through the triangle inequality, we have our result.
\end{proof}

\section{Impossibility of Straightline QNICs}
\label{sec:impossibility}

Now that we have a formal definition settled, we can state a precise version of \cref{thm:main-result}.
\setcounter{theorem}{0}
\begin{theorem}\label{thm:decision}
  Assume there exist quantum commitments relative to the unitary oracle distribution $\oracledist$ and let $(\mathcal{T}_P, \mathcal{T}_V)$ be an $(\varepsilon,\delta)$-straightline QNIC with respect to $\oracledist$ such that $\varepsilon(n) + 2\delta(n) < \kappa^* - \kappa$ for $\kappa^* = 1 - \frac{2}{3\sqrt{3}} \approx \frac{3}{5}$ and some $\kappa > 0$.
  Then there exists a $\QMA$-complete language $\lang$ and an efficient algorithm $\mathcal{A}$ such that
  \begin{align*}
    \Pr_{\oracle \gets \oracledist}[\mathcal{A}^{\oracle}(x) = 1] &\geq 2\delta + \kappa - \negl(n) & \text{ for all } x \in \lang \\
    \Pr_{\oracle\gets\oracledist}[\mathcal{A}^{\oracle}(x) = 1] &< 2\delta + \negl(n) & \text{ for all } x \notin \lang
  \end{align*}
  In particular, if $\kappa \geq 1/\poly(n)$ and $\oracledist$ is efficiently simulatable (as in \cref{def:oracle-sim}), then $\QMA = \BQP$.
\end{theorem}

In the rest of this section, we work to prove \cref{thm:decision}.
In \cref{sec:starprotocol}, we describe the family of counterexample $\Xi$-protocols which we use to construct the decision algorithm.
\cref{sec:compiled-notation} lays out the components for the compiled non-interactive argument and defines some additional notation for use in later sections.
In Sections \ref{sec:case-1} and \ref{sec:case-2}, we construct two subchannels $\mathcal{A}_1$ and $\mathcal{A}_2$ which capture the contributions from invalid and valid ciphertexts, respectively.
Finally, \cref{sec:combine-cases} combines these subchannels into the decision algorithm $\mathcal{A}$ and proves its completeness and soundness.

\subsection{The Counterexample Protocol}
\label{sec:starprotocol}

Let $(\{\xi^{x}\}_{x \in \lang}, \prover, \verify)$ be a $\Xi$-protocol for a $\QMA$-complete language $\lang$ with $(1 - \negl(n))$-completeness, $\negl(n)$-soundness, and a SHVZK simulator $\Sim$ --- e.g., a parallel repetition of the $\Xi$-protocol from \cite{BroadbentG22} which exists if we assume quantum commitments.
Let $(\Gen, \Enc, \Dec)$ be a SKQES with QCA security.
In order to align notation between these protocols, let $\reg{M}$ denote the message register of $\xi^{x}$ and $\reg{A}$ denote a ciphertext register so $\Channel{\Enc}{M}{A}$ and $\Channel{\Dec}{A}{M \oplus \bot}$.
Thus $\xi^x \in \States{MS}$ and the prover and verifier are channels $\Channel{\prover}{CS}{Z}$ and $\Channel{\verify}{MCZ}{B}$.
Without loss of generality, we can assume that both $\prover$ and $\verify$ measure the $\reg{C}$ register immediately.
Although it is technically possible to query them in superposition, any coherence will be lost.
Finally, let $\Channel{S}{C}{MZ}$ be a quantum channel corresponding to the $\Xi$-protocol's SHVZK simulator $\Sim$ --- i.e., for a fixed $x$,
\[
  \mathcal{S}(\rho_\reg{C}) = \sum_{c \in \chlspace} \expval{\rho}{c} \otimes \Sim(x,c).
\]
Given these ingredients, we build a family of $\Xi$-protocols $\{(\starxi^x_k, \starprover, \starverify_k)\}$ indexed by a secret key $k$ as follows.
The state $\starxi^x_k = (\Enc_k \otimes \Id_\reg{S})(\xi^{x}) \in \States{AS}$ is updated such that its message register is encrypted under the parameterized key.
The prover $\starprover$ is unchanged.
On input $\rho_\reg{ACZ}$, the counterexample verifier $\starverify_k$ performs the following steps:
\begin{enumerate}
  \item Apply $\Dec_k \otimes \Id_\reg{CZ}$ to $\rho$.
  \item Measure $\reg{M}$ with $\{\Pi^\bot, \Id - \Pi^\bot\}$. Reject if the outcome is $\Pi^\bot$.
  \item Apply $\verify$ and forward its output.
\end{enumerate}
Using \cref{fact:dec-before-after-proj}, we can write this channel as $\starverify_k(\rho_{\reg{ACZ}}) = \verify \circ \Dec_k(\Pi_k \rho \Pi_k) + \Tr\left[(\Id - \Pi_k)\rho\right]\ketbra{0}$.

\subsection{The Compiled Counterexample NIZK}
\label{sec:compiled-notation}

Before we describe our decision algorithm, we must first lay down a few ingredients starting with the QNIC algorithms.
Recall that $(\mathcal{T}_P, \mathcal{T}_V)$ are single-query oracle algorithms where $\mathcal{T}_P$ prepares the non-interactive argument and $\mathcal{T}_V$ verifies it.
Because $\mathcal{T}_V$ makes only one query, it can be decomposed into $\mathcal{T}_{V,1}$ and $\mathcal{T}_{V,2}$ such that $\mathcal{T}_V^{\mathcal{V}} = \mathcal{T}_{V,2}(\Id \otimes \mathcal{V})\mathcal{T}_{V,1}$.
Additionally, for a fixed oracle $\oracle$, let $(T_{V,1}, \ket{0}_\reg{E})$ be a Stinespring dilation of $\mathcal{T}_{V,1}^{\oracle}$ and we denote by $\tvproj = T_{V,1}(\Id \otimes \ketbra{0}_\reg{E})T_{V,1}^\dagger$ the projector onto the subspace of $T_{V,1}$'s image.
Given a fixed oracle $\oracle$ and dilation $T_{V,1}$, we denote the sequential composition of $\mathcal{T}_P$ with $T_{V,1}$ as
\[
  \mathcal{T}_{P+V,1}^{(\cdot)}(\rho_\reg{A}) = T_{V,1} \left(\mathcal{T}_P^{(\oracle,\cdot)}(\rho) \otimes \ketbra{0}_\reg{E}\right) T_{V,1}^\dagger.
\]
Let $\sigma_k = \mathcal{T}_{P+V,1}^{\starprover}(\starxi^x_k)$ be the state right before $\mathcal{T}_V$'s query to $\starverify_k$ and $\{\accpovm_k, \rejpovm_k\}$ be the measurement for whether that state is both valid (i.e., in the $T_{V,1}$ subspace) and accepted by $\mathcal{T}_V$ --- i.e.,
\[
  \accpovm_k = \tvproj (\mathcal{T}_{V,2} \circ \starverify_k)^\dagger(\!\ketbra{1}) \tvproj, \hspace{1cm} \rejpovm_k = \Id - \accpovm_k
\]

\begin{remark}
  In order to keep notation clear and uncluttered, we will often suppress the parameterization for the oracle $\oracle$.
  So that it is not completely forgotten, we will include $\oracle$ in expectations (e.g., $\E_{\oracle,k}[\ldots]$).
  Keep in mind that all objects downstream from $\mathcal{T}$ will be dependent on $\oracle$ --- e.g., $\tvproj$, $\accpovm_k$, $\sigma_k$, etc.
\end{remark}

Recall the three projectors defined in \cref{lem:canonical} and \cref{def:qca}: $\Pi_k$ which projects onto the space of valid ciphertexts encrypted with the key $k$, $\Pi_{k,r}$ which narrows this to ciphertexts which were specifically encrypted with randomness $r$, and $\widetilde{\Pi}_{k,r}$ which additionally restricts onto those ciphertexts which encrypt half of the maximally entangled state $\ket{\Phi^+}$.
For each $x \in \lang$ and $k$, we can define the subnormalized states
\[
   \Pi_k \sigma_k \Pi_k, \hspace{5mm} \text{and} \hspace{5mm}
  (\Id - \Pi_k) \sigma_k (\Id - \Pi_k)
\]
which correspond to the projection of $x$'s honest proof onto valid and invalid ciphertexts at the point where $\starverify_k$ is decrypting.

Because $\starverify_k(\sigma_k) = \starverify_k(\Pi_k\sigma_k\Pi_k) + \starverify_k((\Id - \Pi_k)\sigma_k(\Id - \Pi_k))$, the expected completeness for a particular statement $x$ on average over all protocols in the family can be written
\ifllncs
\begin{align}
  &\E_{\oracle,k} \Tr[\ketbra{1} \mathcal{T}_V^{\starverify_k} \circ \mathcal{T}_P^{\starprover}(\starxi^x_k)] \nonumber \\
  &\qquad= \E_{\oracle,k} \Tr[\ketbra{1} \mathcal{T}_{V,2}(\starverify_k(\Pi_k\sigma_k\Pi_k) + \starverify_k((\Id - \Pi_k)\sigma_k(\Id - \Pi_k)))] \nonumber \\
  &\qquad= \overbrace{\E_{\oracle,k} \Tr[\ketbra{1} \mathcal{T}_{V,2} \circ \starverify_k(\Pi_k\sigma_k\Pi_k)]}^{\mu^x_1} \nonumber \\
  &\qquad\qquad+ \underbrace{\E_{\oracle,k} \Tr[\ketbra{1} \mathcal{T}_{V,2} \circ \starverify_k((\Id - \Pi_k)\sigma_k(\Id - \Pi_k)))]}_{\mu^x_2}
  \label{eq:mu-defn}
\end{align}
\else
\begin{align}
  \E_{\oracle,k} \Tr[\ketbra{1} \mathcal{T}_V^{\starverify_k} \circ \mathcal{T}_P^{\starprover}(\starxi^x_k)]
  &= \E_{\oracle,k} \Tr[\ketbra{1} \mathcal{T}_{V,2}(\starverify_k(\Pi_k\sigma_k\Pi_k) + \starverify_k((\Id - \Pi_k)\sigma_k(\Id - \Pi_k)))] \nonumber \\
  &= \underbrace{\E_{\oracle,k} \Tr[\ketbra{1} \mathcal{T}_{V,2} \circ \starverify_k(\Pi_k\sigma_k\Pi_k)]}_{\mu^x_1} + \underbrace{\E_{\oracle,k} \Tr[\ketbra{1} \mathcal{T}_{V,2} \circ \starverify_k((\Id - \Pi_k)\sigma_k(\Id - \Pi_k)))]}_{\mu^x_2}
  \label{eq:mu-defn}
\end{align}
\fi
where we've labeled two categories of completeness with $\mu^x_1$ and $\mu^x_2$.
The second category, $\mu^x_2$, captures the event where the non-interactive argument accepts despite the underlying verifier $\starverify_k$ rejecting the ciphertext.
Intuitively, the only way that the compiled argument could be deciding the language is through its interaction with $\starprover$ and $\starverify_k$.
One might therefore expect this value to be small.

Recall that we are assuming there exists $\kappa > 0$ such that $\varepsilon + 2\delta \leq \kappa^* - \kappa$.
In \cref{sec:case-1}, we will construct a subchannel corresponding to the invalid ciphertext completeness contribution $\mu^x_2$, and in \cref{sec:case-2} we do the same for the valid ciphertext contribution $\mu^x_1$.
In \cref{sec:combine-cases}, we combine the two subchannels as complementary branches of a single algorithm.

These subchannels, $\mathcal{A}_1$ and $\mathcal{A}_2$, and their associated hybrids will all output a single classical bit.
As in the previous section, post-selecting onto a projector retains the corresponding measurement outcome without renormalizing: the trace of the output includes the probability of that outcome.
The omitted outcome contributes the zero operator $\mathbf{0}$, not the state $\ketbra{0}$.
For these subchannels, we use the convention
\[
  \Pr_{\oracle}[\mathcal{A}_i^{\oracle}(x)=1]
  = \E_{\oracle}\Tr[\ketbra{1}\mathcal{A}_i^{\oracle}(x)],
\]
and likewise for the hybrids.
Thus all acceptance probabilities below are unconditional weights, not probabilities conditioned on the retained outcome.

\subsection{Case 1: Invalid Ciphertexts} \label{sec:case-1}

In this case, the QNIC algorithms are mauling the ciphertext such that it does not decrypt but they will still accept.
Because the decryption fails, $\mathcal{T}$ will never see the first message in the clear.
We use an encrypted EPR state and retain the branch that fails its authentication test.
QCA security relates this branch to the rejecting branch for an encryption of a fixed state.

\begin{center}
  \fbox{%
    \begin{minipage}{0.9\linewidth}
      $\mathcal{A}_1^{\oracle}(x)$
      \begin{enumerate}
        \item Sample $k \gets \Gen(1^n)$ and $r \gets p_k$.
        \item Prepare $\sigma_{k,r} = \mathcal{T}_{P+V,1}^{\mathcal{S}} \circ \Enc_{k;r}(\ketbra{\Phi^+}_\reg{LM} \otimes \ketbra{x}_\reg{X})$.
        \item Post-select onto $\Id - \widetilde{\Pi}_{k,r}$.
        \item Discard $\reg{ALM'CZ}$, prepare $\reg{B}$ as $\ket{0}$, apply $\mathcal{T}_{V,2}$, and forward the output.
      \end{enumerate}
    \end{minipage}
  }
\end{center}

\begin{lemma} \label{lem:a1-soundness}
  $\Pr_{\oracle}[\mathcal{A}_1^{\oracle}(x) = 1] < \delta(n) + \negl(n)$ for $x \notin \lang$.
\end{lemma}
\begin{proof}
  Fix an oracle $\oracle$.
  Let $\Gamma$ and $\Lambda$ be the real and simulated effective attack maps from \cref{def:qca} for the attack $\mathcal{T}_{P+V,1}^{\mathcal{S}}$, and let $\Gamma^{\mathrm{rej}}$ be the rejecting branch of $\Gamma$.
  By QCA security,
  \[
    \|\Gamma^{\mathrm{rej}}-\ketbra{\bot}_\reg{M}\otimes\Lambda^{\mathrm{rej}}\|_\diamond\leq\negl(n).
  \]
  Let $\mathcal{T}^{\mathrm{acc}}_{V,2}$ be the accepting branch of $\mathcal{T}_{V,2}$ and let $\mathcal{N}$ be the post-processing subchannel which post-selects onto $\Pi^\bot$, discards $\reg{MCZ}$ and any unused simulator output $\reg{M'}$, prepares $\reg{B}$ as $\ket{0}$, and applies $\mathcal{T}^{\mathrm{acc}}_{V,2}$.
  Applying $\mathcal{N}$ to the simulated rejecting branch on input $\ketbra{0,x}$ gives exactly the acceptance weight of $\mathcal{A}_1^{\oracle}(x)$.
  Applying it to the real rejecting branch instead uses an encryption of $\ketbra{0}_\reg{M}$, decrypts after the attack, and post-selects onto $\Pi^\bot$.
  By contractivity, these acceptance weights differ by at most $\negl(n)$.

  Let $\psi_k = \mathcal{T}_P^\mathcal{S}\circ\Enc_k(\ketbra{0}_\reg{M}\otimes\ketbra{x}_\reg{X})$.
  This is an efficiently preparable state because $\mathcal{T}_P$ and $\mathcal{S}$ are efficient, so the compiled argument's soundness applies to $\psi_k$ for every $k$.
  Averaging over $\oracle,k$, we obtain
  \begin{align*}
    \Pr_{\oracle}[\mathcal{A}_1^{\oracle}(x)=1]
    &\leq \E_{\oracle,k}\Tr[\ketbra{1}\mathcal{T}_{V,2}\circ\starverify_k^{\mathrm{rej}}\circ\mathcal{T}_{V,1}(\psi_k\otimes\ketbra{x})]+\negl(n) \\
    &\leq \E_{\oracle,k}\Tr[\ketbra{1}\mathcal{T}_{V,2}\circ\starverify_k\circ\mathcal{T}_{V,1}(\psi_k\otimes\ketbra{x})]+\negl(n) \\
    &< \delta(n) +\negl(n),
  \end{align*}
  where $\starverify_k^{\mathrm{rej}}(\rho)=\Tr[(\Id-\Pi_k)\rho]\ketbra{0}$ and the second inequality follows from $\starverify_k^{\mathrm{rej}}\leq\starverify_k$.
\end{proof}

\begin{lemma} \label{lem:a1-completeness}
  $|\Pr_\oracle[\mathcal{A}_1^{\oracle}(x) = 1] - \mu^x_2| \leq \negl(n)$ for $x \in \lang$.
\end{lemma}
\begin{proof}

This is proven through a series of hybrids such that $\Hyb_2$ accepts with probability $\mu^x_2$.

\begin{center}
\fbox{%
  \begin{minipage}{0.9\linewidth}
    $\Hyb_1^{\oracle}(x)$
    \begin{enumerate}
      \item Sample $k \gets \Gen(1^n)$ and $r \gets p_k$.
      \item Prepare $\mathcal{T}_{P+V,1}^{\changed{\starprover}} \circ \Enc_{k;r}(\ketbra{\Phi^+}_\reg{LM} \otimes \changed{\xi^{x}_\reg{XM'S}})$.
      \item Post-select onto $\Id-\widetilde{\Pi}_{k,r}$. \label{item:case-1-hyb-1-measure}
      \item Discard $\reg{ALM'CZ}$, prepare $\reg{B}$ as $\ket{0}$, apply $\mathcal{T}_{V,2}$, and forward the output. \label{item:case-1-hyb-1-output}
    \end{enumerate}
  \end{minipage}
}
\end{center}

\begin{claim} \label{claim:case-1-hyb-0-vs-1}
  $|\Pr_\oracle[\mathcal{A}_1^\oracle(x) = 1] - \Pr_\oracle[\Hyb_1^{\oracle}(x) = 1]| \leq \negl(n)$ for $x \in \lang$.
\end{claim}
\begin{proof}
  Fix an oracle $\oracle$.
  Observe that the $\reg{M'}$ register is only discarded in Steps \ref{item:case-1-hyb-1-measure} -- \ref{item:case-1-hyb-1-output} and so does not affect the probability that either hybrid outputs 1.
  Thus the only relevant difference between $\mathcal{A}_1$ and $\Hyb_1$ is in the oracle $\mathcal{S}$ and its internal register $\mathsf{S}$.
  Towards contradiction, assume there exists a polynomial $p$ such that
  \[
    |\Pr[\mathcal{A}_1(x) = 1] - \Pr[\Hyb_1(x) = 1]| > 1/p(n).
  \]
  Then we can break the SHVZK of the underlying $\Xi$-protocol.
  The reduction implements the measurement defining $\mathcal{A}_1$ (or $\Hyb_1$), outputs $0$ on the omitted outcome, and otherwise forwards the final output, with $\mathcal{T}_P$'s oracle simulated by the SHVZK challenger.
  That is to say, when $\mathcal{T}_P$ sends $\reg{C}$ to its oracle, measure it in the computational basis, send the outcome $c$ to the challenger, receive $\rho_\reg{M'Z}$ in response, and forward the $\reg{Z}$ register back to $\mathcal{T}_P$.
  If the SHVZK challenger sampled $\rho$ from an honest execution, the reduction accepts with exactly the acceptance weight of $\Hyb_1$.
  Otherwise, its acceptance probability is the acceptance weight of $\mathcal{A}_1$.
  By our assumption, this is an inverse poly distinguisher and we arrive at a contradiction.
\end{proof}

\begin{center}
\fbox{%
  \begin{minipage}{0.9\linewidth}
    $\Hyb_2^{\oracle}(x)$
    \begin{enumerate}
      \item Sample $k \gets \Gen(1^n)$ and $r \gets p_k$.
      \item Prepare $\changed{\Dec_k} \circ \mathcal{T}_{P+V,1}^{\starprover} \circ \Enc_{k;r}(\changed{\xi^x})$.
      \item \changed{Post-select onto $\Pi^\bot$.}
      \item Discard $\reg{MCZ}$, prepare $\reg{B}$ as $\ket{0}$, apply $\mathcal{T}_{V,2}$, and forward the output.
    \end{enumerate}
  \end{minipage}
}
\end{center}

\begin{claim} \label{claim:case-1-hyb-1-vs-2}
  $|\Pr_{\oracle}[\Hyb_1^{\oracle}(x) = 1] - \Pr_{\oracle}[\Hyb_2^{\oracle}(x) = 1]| \leq \negl(n)$ for $x \in \lang$.
\end{claim}
\begin{proof}
  Fix an oracle $\oracle$.
  Let $\Gamma$ (resp.\ $\Lambda$) be the real (resp.\ simulated) effective attack map from \cref{def:qca} for the fixed attack $\mathcal{T}_{P+V,1}^{\starprover}$.
  Let $\mathcal{N}$ be the post-processing subchannel from the proof of \cref{lem:a1-soundness}.
  Then
  \ifllncs
  \begin{align*}
    &|\Pr[\Hyb_2^{\oracle}(x) = 1] - \Pr[\Hyb_1^{\oracle}(x) = 1]| \\
    &\qquad = \left|\Tr\left[\mathcal{N} \circ \Gamma^{\mathrm{rej}}(\xi^x)\right] - \Tr\left[\mathcal{N} \circ (\ketbra{\bot}_\reg{M} \otimes \Lambda^{\mathrm{rej}})(\xi^x)\right]\right| \\
    &\qquad = \left\| \mathcal{N} \circ \Gamma^{\mathrm{rej}}(\xi^x) - \mathcal{N} \circ (\ketbra{\bot}_\reg{M} \otimes \Lambda^{\mathrm{rej}})(\xi^x) \right\|_1 \\
    &\qquad \leq \| \mathcal{N} \circ \Gamma^{\mathrm{rej}} - \mathcal{N} \circ (\ketbra{\bot}_\reg{M} \otimes \Lambda^{\mathrm{rej}}) \|_\diamond \\
    &\qquad \leq \| \Gamma^{\mathrm{rej}} - \ketbra{\bot}_\reg{M} \otimes \Lambda^{\mathrm{rej}} \|_\diamond \\
    &\qquad \leq \negl(n) \qedhere
  \end{align*}
  \else
  \begin{align*}
    |\Pr[\Hyb_2^{\oracle}(x) = 1] - \Pr[\Hyb_1^{\oracle}(x) = 1]| &= \left|\Tr\left[\mathcal{N} \circ \Gamma^{\mathrm{rej}}(\xi^x)\right] - \Tr\left[\mathcal{N} \circ (\ketbra{\bot}_\reg{M} \otimes \Lambda^{\mathrm{rej}})(\xi^x)\right]\right| \\
    &= \left\| \mathcal{N} \circ \Gamma^{\mathrm{rej}}(\xi^x) - \mathcal{N} \circ (\ketbra{\bot}_\reg{M} \otimes \Lambda^{\mathrm{rej}})(\xi^x) \right\|_1 \\
    &\leq \| \mathcal{N} \circ \Gamma^{\mathrm{rej}} - \mathcal{N} \circ (\ketbra{\bot}_\reg{M} \otimes \Lambda^{\mathrm{rej}}) \|_\diamond \\
    &\leq \| \Gamma^{\mathrm{rej}} - \ketbra{\bot}_\reg{M} \otimes \Lambda^{\mathrm{rej}} \|_\diamond \\
    &\leq \negl(n) \qedhere
  \end{align*}
  \fi
\end{proof}
Using \cref{fact:dec-before-after-proj}, we get
\ifllncs
\begin{align*}
  &\Pr_\oracle[\Hyb_2^{\oracle}(x) = 1] \\
  &\hspace{2mm} = \E_{\oracle,k} \Tr\left[\ketbra{1} \mathcal{T}_{V,2}\left(\Tr_\reg{MCZ}\left[\Pi^\bot(\Dec_k \circ \mathcal{T}_{P+V,1}^{\starprover}(\starxi^x_k))\Pi^\bot\right] \otimes \ketbra{0}_\reg{B}\right)\right] \\
  &\hspace{2mm} = \E_{\oracle,k} \Tr\left[\ketbra{1} \mathcal{T}_{V,2}\left(\Tr_\reg{MCZ}\left[\Dec_k((\Id - \Pi_k)(\mathcal{T}_{P+V,1}^{\starprover}(\starxi^x_k))(\Id - \Pi_k))\right] \otimes \ketbra{0}_\reg{B}\right)\right] \\
  &\hspace{2mm} = \E_{\oracle,k} \Tr\left[\ketbra{1} \mathcal{T}_{V,2} \circ \starverify_k((\Id - \Pi_k)\sigma_k(\Id - \Pi_k))\right] \\
  &\hspace{2mm} = \mu^x_2
\end{align*}
\else
\begin{align*}
  \Pr_\oracle[\Hyb_2^{\oracle}(x) = 1]
  &= \E_{\oracle,k} \Tr\left[\ketbra{1} \mathcal{T}_{V,2}\left(\Tr_\reg{MCZ}\left[\Pi^\bot(\Dec_k \circ \mathcal{T}_{P+V,1}^{\starprover}(\starxi^x_k))\Pi^\bot\right] \otimes \ketbra{0}_\reg{B}\right)\right] \\
  &= \E_{\oracle,k} \Tr\left[\ketbra{1} \mathcal{T}_{V,2}\left(\Tr_\reg{MCZ}\left[\Dec_k((\Id - \Pi_k)(\mathcal{T}_{P+V,1}^{\starprover}(\starxi^x_k))(\Id - \Pi_k))\right] \otimes \ketbra{0}_\reg{B}\right)\right] \\
  &= \E_{\oracle,k} \Tr\left[\ketbra{1} \mathcal{T}_{V,2} \circ \starverify_k((\Id - \Pi_k)\sigma_k(\Id - \Pi_k))\right] \\
  &= \mu^x_2
\end{align*}
\fi

\noindent Combining Claims \ref{claim:case-1-hyb-0-vs-1} -- \ref{claim:case-1-hyb-1-vs-2}, we have our result.
\end{proof}

\subsection{Case 2: Valid Ciphertexts} \label{sec:case-2}

On the valid ciphertext branch, we can use our results from \cref{sec:retrospective} to defer the first message until after receiving the challenge.
Because of that reordering, we can replace the prover with the simulator and create what should be a convincing proof without knowledge of the witness.

\begin{center}
\fbox{%
  \begin{minipage}{0.9\linewidth}
    $\mathcal{A}_2^{\oracle}(x)$
    \begin{enumerate}
      \item Sample $k \gets \Gen(1^n)$ and $r \gets p_k$.
      \item Prepare $\sigma_{k,r} = \mathcal{T}_{P+V,1}^{\mathcal{S}} \circ \Enc_{k;r}(\ketbra{\Phi^+}_\reg{LM} \otimes \ketbra{x}_\reg{X})$.
      \item Post-select onto $\widetilde{\Pi}_{k,r}$, then discard $\reg{AL}$.
      \item Encrypt $\reg{M'}$ with $\Enc_{k;r}$ yielding register $\reg{A'}$.
      \item Measure $\reg{A'CZE}$ with $\{\accpovm_k, \rejpovm_k\}$. If the outcome is $\accpovm_k$, output 1. Otherwise output 0.
        \label{step:hyb4-measure}
    \end{enumerate}
  \end{minipage}
}
\end{center}

\begin{lemma}\label{lem:a2-soundness}
  $\Pr[\mathcal{A}_2(x) = 1] < \delta(n)$ for $x \notin \lang$.
\end{lemma}
\begin{proof}
  By the $\delta$-soundness assumption, for all $k$, efficiently preparable states $\psi_k \in \States{P}$, and $x \notin \lang$,
  \[
    \E_\oracle \Tr[(\mathcal{T}_{V,2} \circ \starverify_k)^\dagger(\!\ketbra{1}) T_{V,1} (\psi_k \otimes \ketbra{x}_\reg{X} \otimes \ketbra{0}_\reg{E})T_{V,1}^\dagger] < \delta
  \]
  Towards contradiction, assume $\Pr_{\oracle}[\mathcal{A}_2^{\oracle}(x) = 1] \geq \delta$.
  For fixed $k$, let $\psi_k$ be the subnormalized state before Step \ref{step:hyb4-measure} of $\mathcal{A}_2^{\oracle}$, averaged over $r$; its trace includes the probability of the retained outcome in Step 3.
  Then by our assumption, $\E_{\oracle,k} \Tr[\accpovm_k\psi_k] \geq \delta$.
  By an averaging argument, there exists some $k^*$ such that $\E_\oracle \Tr[\accpovm_{k^*}\psi_{k^*}] \geq \delta$.
  Let $\varphi$ be the state prepared in the following way:
  \begin{enumerate}
    \item Run Steps 1 -- 4 of $\mathcal{A}_2^{\oracle}$ with $k=k^*$, physically measuring $\{\widetilde{\Pi}_{k^*,r},\Id-\widetilde{\Pi}_{k^*,r}\}$ in Step 3. If the outcome is $\Id-\widetilde{\Pi}_{k^*,r}$, output any fixed state on $\reg{P}$ and stop.
    \item Measure with $\{\tvproj,\Id-\tvproj\}$. If the outcome is $\tvproj$, apply $T_{V,1}^\dagger$ and output the residual state on $\reg{P}$.
    \item Otherwise, output any fixed state on $\reg{P}$.
  \end{enumerate}
  This prepares a normalized state efficiently, without conditioning on either measurement outcome.
  If we apply $T_{V,1}$ to $\varphi$ and measure with $(\mathcal{T}_{V,2} \circ \starverify_{k^*})^\dagger(\!\ketbra{1})$, the branch retaining both outcomes contributes exactly $\Tr[\accpovm_{k^*}\psi_{k^*}]$ to acceptance.
  The other branches contribute nonnegative amounts, so the total accepting probability is $\geq\delta$ in expectation over $\oracle$.
  This contradicts $\delta$-soundness and the lemma holds.
\end{proof}

\begin{lemma}\label{lem:a2-completeness}
  $\Pr[\mathcal{A}_2(x) = 1] \geq \mu^x_1 - \frac{2}{3\sqrt{3}} - \negl(n)$ for $x \in \lang$.
\end{lemma}
\begin{proof}

This is proven through a series of subchannel hybrids where $\Hyb_3$ is the trace-preserving special case corresponding to the honest execution.

\begin{center}
\fbox{%
  \begin{minipage}{0.9\linewidth}
    $\Hyb_1^{\oracle}(x)$
    \begin{enumerate}
      \item Sample $k \gets \Gen(1^n)$ and $r \gets p_k$.
      \item Prepare $\sigma_{k,r} = \mathcal{T}_{P+V,1}^{\changed{\starprover}} \circ \Enc_{k;r}(\ketbra{\Phi^+}_\reg{LM} \otimes \changed{\xi^x_\reg{XM'S}})$.
      \item Post-select onto $\widetilde{\Pi}_{k,r}$, then discard $\reg{AL}$.
      \item Encrypt $\reg{M'}$ with $\Enc_{k;r}$ yielding register $\reg{A'}$. \label{item:case-2-hyb-1-before-measure}
      \item Measure $\reg{A'CZE}$ with $\{\accpovm_k, \rejpovm_k\}$. If the outcome is $\accpovm_k$, output 1. Otherwise output 0.
    \end{enumerate}
  \end{minipage}
}
\end{center}
\begin{claim} \label{claim:case-2-hyb-0-vs-1}
  $|\Pr_{\oracle}[\mathcal{A}_2^{\oracle}(x) = 1] - \Pr_{\oracle}[\Hyb_1^{\oracle}(x) = 1]| \leq \negl(n)$ for $x \in \lang$.
\end{claim}
\begin{proof}
  By the SHVZK argument in \cref{claim:case-1-hyb-0-vs-1}, with the reduction outputting $0$ on the omitted measurement outcome. Thus it compares the unconditional acceptance weights of the two subchannels.
\end{proof}

\begin{center}
\fbox{%
  \begin{minipage}{0.9\linewidth}
    $\Hyb_2^{\oracle}(x)$
    \begin{enumerate}
      \item Sample $k \gets \Gen(1^n)$ and $r \gets p_k$.
      \item Prepare $\sigma_{k,r} = \mathcal{T}_{P+V,1}^{\starprover} \circ \Enc_{k;r}(\changed{\xi^x_\reg{XMS}})$.
      \item \changed{Post-select onto $\Pi_k$.} \label{item:case-2-hyb-2-before-measure}
      \item Measure $\reg{\changed{A}CZE}$ with $\{\accpovm_k, \rejpovm_k\}$. If the outcome is $\accpovm_k$, output 1. Otherwise output 0.
    \end{enumerate}
  \end{minipage}
}
\end{center}

\begin{claim}
  $|\Pr_{\oracle}[\Hyb_1^{\oracle}(x) = 1] - \Pr_{\oracle}[\Hyb_2^{\oracle}(x) = 1]| \leq \negl(n)$ for $x \in \lang$.
\end{claim}
\begin{proof}
  Fix an oracle $\oracle$.
  Steps 1 -- \ref{item:case-2-hyb-1-before-measure} of $\Hyb_1^{\oracle}$ are exactly the steps in $\Exp_\mathsf{defer}^{\mathcal{A}}(\xi^x)$ for $\mathcal{A} = \mathcal{T}_{P+V,1}^{\starprover}$ and similarly Steps 1 -- \ref{item:case-2-hyb-2-before-measure} of $\Hyb_2^{\oracle}$ correspond to $\Exp_\mathsf{real}^\mathcal{A}(\xi^x)$.
  By \cref{lem:exps-close}, these subchannels are negligibly close in diamond distance\footnote{Note that both experiments output $k$ so the measurement can depend on $k$.}.
  Therefore, the final measurement accepts in the two hybrids with probabilities differing by at most $\negl(n)$.
\end{proof}

\begin{center}
\fbox{%
  \begin{minipage}{0.9\linewidth}
    $\Hyb_3^{\oracle}(x)$
    \begin{enumerate}
      \item Sample $k \gets \Gen(1^n)$ and $r \gets p_k$.
      \item Prepare $\sigma_{k,r} = \mathcal{T}_{P+V,1}^{\starprover} \circ \Enc_{k;r}(\xi^x_\reg{XMS})$.
      \stepcounter{enumi}
      \item[\changed{\labelenumi}] Measure $\reg{ACZE}$ with $\{\accpovm_k, \rejpovm_k\}$. If the outcome is $\accpovm_k$, output 1. Otherwise output 0.
    \end{enumerate}
  \end{minipage}
}
\end{center}

\begin{claim} \label{claim:case-2-hyb-2-vs-3}
  $\Pr_{\oracle}[\Hyb_3^{\oracle}(x) = 1] - \Pr_{\oracle}[\Hyb_2^{\oracle}(x) = 1] \leq \mu^x_2 + \frac{2}{3\sqrt{3}}$ for $x \in \lang$.
\end{claim}
\begin{proof}
  Fix $\oracle,k,r$.
  By construction, $\tvproj\sigma_{k,r}\tvproj = \sigma_{k,r}$.
  Moreover, $(\mathcal{T}_{V,2} \circ \starverify_k)^\dagger(\!\ketbra{1})$ commutes with $\Pi_k$, since $\starverify_k$ measures ciphertext validity before proceeding.
  Apply \cref{lem:alternating-projections} with projectors $\Pi_k,\tvproj$, measurement operator $(\mathcal{T}_{V,2} \circ \starverify_k)^\dagger(\!\ketbra{1})$, and state $\sigma_{k,r}$ to obtain
  \[
    \Tr[\ketbra{1}\mathcal{T}_{V,2}\circ\starverify_k(\Pi_k\sigma_{k,r}\Pi_k)] - \Tr[\accpovm_k\Pi_k\sigma_{k,r}\Pi_k] \leq \frac{2}{3\sqrt{3}}.
  \]
  Taking the expectation over $(\oracle, k, r)$ and using $\E_r[\sigma_{k,r}] = \sigma_k$ gives
  \[
    \mu^x_1 - \Pr_{\oracle}[\Hyb_2^{\oracle}(x)=1]
    \leq \frac{2}{3\sqrt{3}}.
  \]
  Since $\Pr_{\oracle}[\Hyb_3^{\oracle}(x)=1]=\mu^x_1+\mu^x_2$, the claim follows.
\end{proof}

$\Hyb_3^{\oracle}$ is equivalent to the honest execution of the compiled non-interactive argument for $x$ with a protocol randomly sampled from our family.
Combining Claims \ref{claim:case-2-hyb-0-vs-1} -- \ref{claim:case-2-hyb-2-vs-3}, we get
\[
  \Pr_{\oracle}[\mathcal{A}_2^{\oracle}(x) = 1]
  \geq \Pr_{\oracle}[\Hyb_3^{\oracle}(x)=1] - \mu^x_2 - \frac{2}{3\sqrt{3}} - \negl(n) \\
  = \mu^x_1 - \frac{2}{3\sqrt{3}} - \negl(n). \qedhere
\]
\end{proof}

\subsection{Proof of Theorem 1} \label{sec:combine-cases}

The subchannels $\mathcal{A}_1$ and $\mathcal{A}_2$ both begin by preparing $\sigma_{k,r}$.
They then retain complementary outcomes of the same projective measurement, and their subsequent steps are trace preserving.
Consequently, for every fixed $(\oracle,x)$,
\[
  \Tr[\mathcal{A}_1^{\oracle}(x)]+\Tr[\mathcal{A}_2^{\oracle}(x)]
  =\E_{k,r}\Tr[(\Id-\widetilde{\Pi}_{k,r})\sigma_{k,r}]
  +\E_{k,r}\Tr[\widetilde{\Pi}_{k,r}\sigma_{k,r}]
  =1.
\]
Thus their sum defines a channel $\mathcal{A}^{\oracle}=\mathcal{A}_1^{\oracle}+\mathcal{A}_2^{\oracle}$.
Operationally, $\mathcal{A}$ performs the common preparation once and measures $\{\widetilde{\Pi}_{k,r},\Id-\widetilde{\Pi}_{k,r}\}$.
On outcome $\Id-\widetilde{\Pi}_{k,r}$, it continues as in $\mathcal{A}_1$; on outcome $\widetilde{\Pi}_{k,r}$, it continues as in $\mathcal{A}_2$.
It always forwards the resulting output bit, without post-selection.
In particular, for all $x$, the acceptance weights add exactly:
\begin{equation} \label{eq:combined-acceptance}
  \Pr_{\oracle}[\mathcal{A}^{\oracle}(x)=1]
  =\Pr_{\oracle}[\mathcal{A}_1^{\oracle}(x)=1]
  +\Pr_{\oracle}[\mathcal{A}_2^{\oracle}(x)=1].
\end{equation}
All steps are efficient, so $\mathcal{A}$ is itself efficient.

\begin{lemma}
  $\Pr_{\oracle}[\mathcal{A}^{\oracle}(x) = 1] \leq 2\delta + \negl(n)$ if $x \notin \lang$.
\end{lemma}
\begin{proof}
  By Lemmas \ref{lem:a1-soundness} and \ref{lem:a2-soundness}, each subchannel has acceptance weight $<\delta+\negl(n)$.
  The result follows from \cref{eq:combined-acceptance}.
\end{proof}

\begin{lemma}
  $\Pr_{\oracle}[\mathcal{A}^{\oracle}(x) = 1] \geq 2\delta + \kappa - \negl(n)$ if $x \in \lang$.
\end{lemma}
\begin{proof}
  By \cref{eq:combined-acceptance} and Lemmas \ref{lem:a1-completeness} and \ref{lem:a2-completeness},
  \begin{align*}
    \Pr_{\oracle}[\mathcal{A}^{\oracle}(x)=1]
    &\geq \mu^x_2+\mu^x_1-\frac{2}{3\sqrt{3}}-\negl(n) \\
    &\geq 1-\varepsilon-\frac{2}{3\sqrt{3}}-\negl(n) \\
    &= \kappa^*-\varepsilon-\negl(n) \\
    &\geq 2\delta+\kappa-\negl(n).
  \end{align*}
  where the second inequality uses our completeness bound and the last uses the hypothesis $\varepsilon+2\delta<\kappa^*-\kappa$.
\end{proof}

\bibliographystyle{alpha}
\bibliography{references}

\appendix

\end{document}